\documentclass[12pt,reqno]{article}
\usepackage{fullpage} 
\usepackage{amssymb} 
\usepackage{graphicx} 
\usepackage{amsmath} 
\usepackage{amsthm,amssymb,latexsym} 
\usepackage{amstext, amsfonts,amsbsy} 
\usepackage{enumerate} 
\usepackage{upgreek} 
\usepackage{color} 
\usepackage{accents} 
\usepackage{mathtools} 
\usepackage{float} 
\usepackage{tikz}
\usetikzlibrary{matrix,graphs,arrows,positioning,calc,decorations.markings,shapes.symbols}
\usepackage[pdftex,bookmarks,colorlinks,breaklinks]{hyperref}  
\definecolor{dullmagenta}{rgb}{0.4,0,0.4}   
\definecolor{darkblue}{rgb}{0,0,0.4}
\hypersetup{linkcolor=red,citecolor=blue,filecolor=dullmagenta,urlcolor=darkblue} 
\usepackage{eucal}

\newtheorem{theorem}{Theorem}

\newtheorem{lemma}[theorem]{Lemma}

\newcommand{\Pain}[1]{\text{P}_{\mathrm{#1}}}
\newcommand{\dPain}[1]{\text{P}\left(\mathrm{#1}\right)}

\theoremstyle{definition}

\theoremstyle{remark}
\newtheorem{remark}[theorem]{Remark}

\newtheorem*{notation*}{Notation} 

\begin{document}

{\noindent\Large\bf Recurrence relations for the generalized Laguerre and Charlier orthogonal polynomials and discrete Painlev\'e equations on the $D_{6}^{(1)}$ 
Sakai surface}
\medskip
\begin{flushleft}

\textbf{Xing Li}\\
School of Mathematics and Statistics, Jiangsu Normal University,\\ Xuzhou 221116,  P.R.China\\
E-mail: \href{mailto:xing_li@jsnu.edu.cn}{\texttt{xing\_li@jsnu.edu.cn}}\\[5pt] 

\textbf{Anton Dzhamay}\\
Beijing Institute of Mathematical Sciences and Applications (BIMSA), No. 544,\\ Hefangkou Village, 
Huaibei Town, Huairou District, Beijing 101408, P.R.China\\
E-mail: \href{mailto:adzham@bimsa.cn}{\texttt{adzham@bimsa.cn}}\\[5pt]

\textbf{Galina Filipuk}\\
Institute of Mathematics, 
University of Warsaw, Banacha 2, Warsaw, 02-097, Poland\\
E-mail: \href{mailto:g.filipuk@uw.edu.pl}{\texttt{g.filipuk@uw.edu.pl}}\\[5pt] 

\textbf{Da-jun Zhang}\\
	Department of Mathematics, Shanghai  University and Newtouch Center for Mathematics, \\ 
Shanghai University, Shanghai 200444,  P.R.China\\
E-mail: \href{djzhang@staff.shu.edu.cn}{\texttt{djzhang@staff.shu.edu.cn}}\\[5pt] 



\end{flushleft}

\date{\today}

\begin{abstract}
	
This paper concerns the discrete version of the \emph{Painlev\'e identification problem}, i.e., how to recognize a certain 
recurrence relation as a discrete Painlev\'e equation. Often some clues can be seen from the setting of the problem, e.g., when
the recurrence is connected with some differential Painlev\'e equation, or from the geometry of the configuration of 
indeterminate points of the equation. 
The main message of our paper is that, in fact, this only allows us to identify the \emph{configuration space}  of the dynamic system, 
but not the dynamics themselves. The  \emph{refined version} of the identification problem lies in determining, up to the conjugation, 
the translation direction of the dynamics, which in turn requires the full power of the geometric theory of 
Painlev\'e equations. 

To illustrate this point, in this paper we consider two examples of such recurrences that appear in the theory of 
orthogonal polynomials. We choose these examples because they get regularized on the same family of Sakai surfaces, 
but at the same time are not equivalent, since they result in non-equivalent translation directions. In addition, 
we show the effectiveness of a recently proposed identification procedure for discrete Painlev\'e equations 
using Sakai's geometric approach for answering such questions. In particular, this approach requires no a priori knowledge of 
a possible type of the equation.


\end{abstract}

\section{Introduction} 
\label{sec:introduction}

It is well known that Painlev\'e equations, both differential and discrete, often appear in the study of various orthogonal
polynomial ensembles, see a monograph of Walter Van Assche, \cite{Van:2018:OPPE}, as well as a more recent survey 
\cite{Van:2022:OPTLPE}. For example, 
discrete Painlev\'e equations describe the dependence of the coefficients of the \emph{three-term recurrence relations} and 
\emph{ladder operators} as
functions of the degree variable $n$, whereas differential Painlev\'e equations sometimes describe their dependence on 
some continuous weight parameters. This relationship is due to a connection between orthogonal polynomial ensembles and
certain isomonodromy problems, 
 which goes back to the seminal paper \cite{FIK:1991:CMP}, where the notion of a discrete Painlev\'e
equation explicitly appeared for the first time, see also \cite{PapNijGraRam:1992:IDPDAPE}. In the discrete case, an important contribution was made
by Alexei Borodin, who developed the \emph{Discrete Riemann-Hilbert Problem} formalism to 
study correlation functions of integrable operators \cite{Bor:2003:DPDPE}, extending the Riemann-Hilbert approach 
developed by Percy Deift in the Random Matrix Theory context, see \cite{BorDei:2002:FDJTRT}. In the series of further works
Borodin and his collaborators rephrased that in terms of isomonodromic transformations of $d$-connections, which 
gave a further geometric explanation of appearance of discrete Painlev\'e equations in a large class of models in 
integrable probability, \cite{BorBoy:2003:DFPDOPE,AriBor:2006:MSDDPE,AriBor:2009:TDITP}.
However, in practice it is often quite difficult
to recognize exactly which Painlev\'e equations appear. And even when we have some problem-specific information that suggests
the type of the corresponding Painlev\'e equations, finding explicit reduction to some canonical form is a highly non-obvious
task. In a recent paper \cite{DzhFilSto:2020:RCDOPWHWDPE} two of the authors, together with Alexander Stokes,
suggested an explicit step-by-step procedure of reducing discrete Painlev\'e equations that appear in applied problems to 
their standard form by finding some appropriate change of coordinates. This procedure is based on the geometric 
theory of discrete Painlev\'e equations created by Hidetaka Sakai  \cite{Sak:2001:RSAWARSGPE}, see also an important recent survey 
paper \cite{KajNouYam:2017:GAPE} and many references therein.  In general, Sakai's geometric theory seems to be a very natural framework for this kind of problems. 

Another interesting aspect of this problem is that, in addition to the geometric classification of discrete Painlev\'e equations by the type of 
the rational algebraic surface on which the equation is regularized, there is a finer classification by the conjugacy class of the translation 
element in the extended affine Weyl symmetry group of the surface, where this translation element defined the actual dynamic. And even within
the same translation class it may be possible to have additional constraints on point configurations resulting in new properties of the resulting
dynamic. Thus, among this manifold of examples, it is interesting to see which ones appear in applications, such as the 
study of orthogonal polynomials. In that context, it is also interesting to understand the relationship between equations that correspond to 
different weights, as well as consider various limits in the weight parameters. 
The geometric approach is particularly well-suited for answering such questions. To better understand the refined version of the discrete
case of the Painlev\'e identification problem, we need to have a large library of concrete examples, where for each recurrence we find not
only the conjugation class of the dynamic, but also the relationship between the  parameters appearing in an applied problem, such as
the weight parameters, and the canonical parameters for the surface, the so-called \emph{root variables}. 
Thus, the present paper can also be considered as one of the steps towards constructing a large library of such algebro-geometric data.

%
%

The two examples we consider in this paper are from the theory of orthogonal polynomials. 
The first recurrence describes the coefficients of ladder operators for a certain generalization 
of Laguerre orthogonal polynomials, and the second describes the coefficients of the three-term recurrence 
relation for generalized Charlier polynomials. Both recurrences get regularized on a rational algebraic surface of type $D_{6}^{(1)}$, but the resulting 
translations are \emph{non-equivalent}. In fact, the first recurrence corresponds to discrete Painlev\'e equation given by 
equation (8.29) in Section 8.1.20 in \cite{KajNouYam:2017:GAPE} (equation \eqref{eq:dP-KNY} in Appendix~\ref{sec:dP-2A1-std}), 
and the second to the so-called alt.~d-$\Pain{II}$ equation in \cite{Sak:2001:RSAWARSGPE} (equation \eqref{eq:dP-Sakai-map} in Appendix~\ref{sec:dP-2A1-std}).
Based on the translations these equations define, we denote them by $[1 \overline{1} \overline{1}1]$  and by $[00 \overline{1}1]$ 
respectively.


The paper is organized as follows. In Section~\ref{sec:Laguerre-weight} we give a rather detailed step-by-step illustration on how to use 
the geometric tools of the Sakai theory to solve the Painlev\'e identification problem for the generalized Laguerre weight example. 
In Section~\ref{sec:Charlier-weight} we give a very brief summary of the 
identification for the generalized Charlier example. We collect some of the standard algebro-geometric data for the $D_{6}^{(1)}$ surface 
in Appendix~\ref{sec:dP-2A1-std}.


\section{Generalized Laguerre Weight and the MacDonald Hierarchy} 
\label{sec:Laguerre-weight}

In this section we consider an example of a recurrence relation on the coefficients of ladder operators that appeared in a 
paper \cite{CheIts:2010:PSLSHRME} by Yang Chen and Alexander Its, who studied a certain linear statistics
for the \emph{unitary Laguerre ensemble}. These statistics can in turn be interpreted as a singular perturbation of the 
standard  Laguerre weight $x^{\alpha} e^{-x}$ by a multiplicative factor $e^{-s/x}$ that induces an 
infinitely strong zero at the origin. This perturbed Laguerre weight was further considered in 
\cite{CheChe:2015:SLSLUEPIDSA,CheFilReb:2019:NDEMLWLEA,XuDaiZha:2015:PAHDSPLW} (see also the references therein).

The setup is as follows. Consider a collection of monic polynomials
$\{P_{n}(x) = x^{n} + p_{1}(n) x^{n-1} + \cdots + P_{n}(0)\}$ 
orthogonal with respect to the weight
$w(x) = w(x; \alpha,s)$, depending on two parameters $\alpha>0, s>0$ 
(we often omit the explicit dependence of parameters or include only the relevant one, usually $s$, in the formulae below):
\begin{equation}
	\int_{0}^{\infty} P_{j}(x) P_{k}(x) w(x; \alpha,s)\, dx = h_{j}(\alpha,s) \delta_{jk},\qquad 
	\begin{aligned}
		&w(x; \alpha,s) = x^{\alpha} e^{-x} e^{-s/x},\\
		&0\leq x < \infty,\  \alpha>0, s>0.
	\end{aligned}
\end{equation}
Orthogonal polynomials satisfy a three-term recurrence relation 
\begin{equation}\label{eq:three-term}
x P_{n}(x) = P_{n+1}(x) + \alpha_{n} P_{n}(x) + \beta_{n} P_{n-1}(x),
\end{equation}
where $\alpha_{n}\in \mathbb{R}$ and $\beta_{n} = h_{n}/h_{n-1}$.
One can also define the so-called \emph{ladder operators}
\begin{align*}
	\left( \frac{d}{dx} + B_{n}(x) \right) P_{n}(x) &= \beta_{n} A_{n}(x) P_{n-1}(x),\\
	\left( \frac{d}{dx} - B_{n}(x) - v'(x) \right) P_{n-1}(x) &= - A_{n-1}(x) P_{n}(x),\qquad v(x) := - \ln w(x),
\end{align*}
satisfying the fundamental compatibility conditions 
\begin{align*}
	B_{n+1}(x) + B_{n}(x) &= (x - \alpha_{n}) A_{n}(x) - v'(x),\\
	1 + (x-\alpha_{n})(B_{n+1}(x) - B_{n}(x)) &= \beta_{n+1} A_{n+1}(x) - \beta_{n} A_{n-1}(x).
\end{align*}
In this particular case the ladder operators can be parameterized by some variables $a_{n}$ and $b_{n}$ as follows:
\begin{alignat*}{2}
	A_{n}(x) &= \frac{1}{x} + \frac{a_{n}}{x^{2}}, \qquad & 
	a_{n}(s) &= \frac{s}{h_{n}} \int_{0}^{\infty} \frac{P_{n}^{2}(x)}{x} w(x)\, dx,\\
	B_{n}(x) &= -\frac{n}{x} + \frac{b_{n}}{x^{2}}, \qquad & 
	b_{n}(s) &= \frac{s}{h_{n-1}} \int_{0}^{\infty} \frac{P_{n}(x)P_{n-1}(x)}{x} w(x)\, dx.
\end{alignat*}
These variables satisfy the recurrence relation
\begin{equation}
\begin{aligned}\label{eq:recurr}
b_n+b_{n+1}&=s-(2n+1+\alpha+a_{n})a_{n},\\
(b_{n}^2-s b_{n})(a_{n}+a_{n-1})&=[n s-(2n+\alpha)b_{n}]a_{n} a_{n-1},
\end{aligned}
\end{equation}
which is the recurrence relation that we study. For the orthogonal polynomial ensemble under 
consideration the initial conditions for this recurrence are given by
\begin{equation*}
a_0(s)=\sqrt{s}\frac{K_\alpha(2\sqrt{s})}{K_{\alpha+1}(2\sqrt{s})},\qquad b_0(s)=0,
\end{equation*}
where $K_{\alpha}(z)$ is the MacDonald function (the modified Bessel function) of the second kind, and for this 
reason this recurrence is called the MacDonald hierarchy in \cite{CheIts:2010:PSLSHRME}. Note that it is also possible
to explicitly express the coefficients $\alpha_{n}$ and $\beta_{n}$ of the three-term recurrence relation in terms of 
variables $a_{n}$ and $b_{n}$,
\begin{equation*}
	\alpha_{n} = 2n + 1 + \alpha + a_{n},\qquad \beta_{n} a_{n}^{2} = [ns - (2n + \alpha)b_{n}]a_{n} - (b_{n}^{2} - s b_{n}).
\end{equation*}

Another important objects associated with the weight $w(x;\alpha,s)$ is the determinant $D_{n}(s)$ of its moment matrix, also known as the 
\emph{Hankel determinant}, and its logarithmic derivative $H_{n}(s)$,
\begin{equation}
D_n(s)=\det\left[\int_0^\infty x^{j+k}x^\alpha e^{-x-s/x}\right]_{j,k=0}^{n-1} = \prod_{j=0}^{n-1} h_{j}(\alpha,s),\qquad 
H_{n}(s):= s \frac{d}{ds}\ln D_{n}(s).
\end{equation}
As shown in \cite{CheIts:2010:PSLSHRME}, the quantities $a_{n}$ and $b_{n}$ (and hence   $\alpha_{n}$ and $\beta_{n}$) can be written in terms of this 
logarithmic derivative as follows:
\begin{equation}
	\begin{aligned}
		a_{n}&=H_{n}-H_{n+1},\\
		b_{n}&=\frac{ns+\delta^{2}H_{n}[H_{n}-n(n+\alpha)]}{2n+\alpha+\delta^2H_{n}},\quad \text{where }\delta^{2}H_{n} := H_{n-1}-H_{n+1} = a_{n}+a_{n-1}.
	\end{aligned}
\end{equation}
Substituting the above into \eqref{eq:recurr} results 
\begin{equation}\label{eq:Hn}
	\begin{aligned}
		&\{[H_n-n(n+\alpha)]\delta^2H_n+ns\}\{[H_n-n(n+\alpha)-s]\delta^2H_n-(n+\alpha)s\} \\
		&\qquad  =(2n+\alpha+\delta^2H_n)\{ns+(2n+\alpha)[n(n+\alpha)-H_n]\}(H_n-H_{n+1})(H_{n-1}-H_n).			
	\end{aligned}
\end{equation}
It is not difficult to show that conversely, equation \eqref{eq:Hn} is equivalent to recurrence relations \eqref{eq:recurr}.

In \cite{CheIts:2010:PSLSHRME} Chen and Its gave the Lax Pair for equation \eqref{eq:Hn}, and so this is an 
integrable discrete equation. Moreover, in \cite{CheIts:2010:PSLSHRME,CheFilReb:2019:NDEMLWLEA} it has been argued that it can be obtained as some
composition of elementary B\"acklund-Schlesinger transformation of the (generic) third Painlev\'e equation, and hence it should be
one of discrete Painlev\'e equations on the $D_{6}^{(1)}$ Sakai surface. However, an explicit identification with one of the 
standard discrete Painlev\'e equations was not provided. In this section we show that 
\eqref{eq:recurr} is indeed equivalent to the standard  d-$\dPain{2A_1^{(1)}/D_6^{(1)}}$ equation 
as written in \cite{KajNouYam:2017:GAPE}. However, we first change the notation from $a_{n}$ 
and $b_{n}$ to $x_{n}$ and $y_{n}$ to avoid clashing with the standard notation for root variables in the 
geometric Painlev\'e theory. Thus, our main object of study is the following discrete dynamical system:
\begin{equation}\label{eq:Laguerre-xy-rec}
\left\{\begin{aligned}
(x_{n}+x_{n-1}) (y_{n}^2-s y_{n}) &=\big(n s-(2n+\alpha)y_{n}\big)x_{n} x_{n-1},\\
y_n+y_{n+1}&=s-(2n+1+\alpha+x_{n})x_{n}.
\end{aligned}\right.
\end{equation}
Our main result is the explicit change of variables and parameter identification establishing the
above equivalence that is given in the following Theorem.

\begin{theorem}\label{thm:Laguerre-coord-change}
	Recurrence \eqref{eq:Laguerre-xy-rec} is equivalent to the standard d-P$(2A_1^{(1)}/D_6^{(1)})$ equation \eqref{eq:dP-KNY},
	which we reproduce here for the convenience of the reader,
	\begin{equation*}
	q_{n+1}+q_{n}=-\frac{a_2}{p_{n}}-\frac{a_1}{p_{n}-1},\qquad p_{n}+ p_{n-1} = 1+\frac{1-a_{2}-a_{1}}{q_{n}}-\frac{t}{q_{n}^2},
	\end{equation*}	
	via the following explicit change of coordinates and parameter identification (here $a_{i}$ denote the \emph{root variables},
	or the canonical parameters for a discrete Painlev\'e equation,
	see Section~\ref{sub:Laguerre-KNY-pars})
	\begin{equation}\label{eq:Laguerre-KNY-coord-change}
		\left\{
		\begin{aligned}
			x_{n}(q,p) &= \frac{t}{q_{n}},\\
			y_{n}(q,p) &= -p_{n}t - \frac{t(a_{2}-a_{0})}{q_{n}} - \frac{t^{2}}{q_{n}^{2}} \\ &= t(p_{n-1} -1 ),\\
			n&=a_{1} - 1=-a_{0},\\  \alpha &= a_{2} - a_{1},\ s=-t,			
		\end{aligned}
		\right. 
		\quad 
		\left\{
		\begin{aligned}
		q_{n}(x,y) &= -\frac{s}{x_{n}},\\
		p_{n}(x,y) &= \frac{x_{n}^{2} +x_{n}(1 + 2n + \alpha)+ y_{n}}{s} \\ &= 1 - \frac{y_{n+1}}{s},\\
		a_0&=-n,\ a_{1}=1+n,\\ 
		a_2&=1 + n+\alpha,\ a_{3}=-n-\alpha,\  t=-s.			
		\end{aligned}
		\right.
	\end{equation}
\end{theorem}

\begin{remark} This result can, after the fact,  be established by a direct observation, see \cite{Van:2022:OPTLPE}. However, that requires
	knowing the type of the surface, as well as the equation that we are expected to match. The advantage of the geometric approach is that
	it requires no prior knowledge of the resulting equation. And the knowledge of just the surface type of the equation is not sufficient,
	as we show in the second example below.
\end{remark}

We now give a detailed proof of this theorem following the algorithmic approach of \cite{DzhFilSto:2020:RCDOPWHWDPE}.

\subsection{The singularity structure of the recurrence relation} 
\label{sub:the_singularity_structure_of_the_recurrence_relation}

Recurrence \eqref{eq:Laguerre-xy-rec} defines two half-step mappings,  
the \emph{forward half-mapping} 
\begin{equation}\label{eq:Lag-ph1}
	\tilde{\varphi}_{1}^{n}:(x_{n},y_{n})\mapsto (x_{n},y_{n+1}) = \left(x_{n},s - y_{n} - x_{n}( 1 + 2 n + x_{n} + \alpha)\right),
\end{equation}
and the \emph{backward half-mapping}
\begin{equation}\label{eq:Lag-ph2}
	\tilde{\varphi}_{2}^{n}:(x_{n},y_{n})\mapsto (x_{n-1},y_{n}) = 
		\left( \frac{x_{n}y_{n}(y_{n} - s)}{n x_{n} (s - 2 y_{n}) - y_{n}(\alpha x_{n} + y_{n} - s)} ,y_{n}\right).
\end{equation}

The first step of the process is to extend the mapping from $\mathbb{C} \times \mathbb{C}$ to $\mathbb{P}^{1} \times \mathbb{P}^{1}$ by 
introducing three additional charts $(X,y)$, $(x,Y)$, and $(X,Y)$, where $X = 1/x$, $Y = 1/y$, and 
then to find the \emph{base points} of the mapping where the numerator and the denominator of rational functions defining the mapping 
simultaneously vanish. It is also convenient to omit the index $n$ by introducing the standard notation 
$x:= x_{n}$, $\overline{x}:=x_{n+1}$, $\underline{x}:=x_{n-1}$, and similarly for $y$.

Looking at the backward half-mapping in the $(x,y)$-chart,
\begin{equation*}
	\tilde{\varphi}_{2}:(x,y)\mapsto (\underline{x},y) = 
		\left( \frac{xy(y - s)}{n x (s - 2 y) - y(\alpha x+y - s)} ,y\right),
\end{equation*}
we immediately see two base points, $q_{1}(0,0)$ and $q_{3}(0,s)$ (the reason for the choice for indices will be clear once we find all of the base points)
for the mapping $\underline{x} = \underline{x}(x,y)$.
To resolve the indeterminacy of the mapping, for each base point $q_{i}(x_{i},y_{i})$ we introduce two new local 
coordinate charts $(u_{i},v_{i})$ and $(U_{i},V_{i})$ by
\begin{equation*}
	x = x_{i} + u_{i} = x_{i} + U_{i} V_{i},\quad y = y_{i} + u_{i} v_{i} = y_{i} + V_{i}. 
\end{equation*}
This is known as the \emph{blowup procedure} in algebraic geometry, see standard textbooks \cite{Sha:2013:BAG1,GriHar:1978:POAG}
(or \cite{DzhFilSto:2020:RCDOPWHWDPE} for a brief summary sufficient for our purposes). This procedure modifies the geometry of the domain (and hence the 
range) of the mapping by gluing an \emph{exceptional divisor} $F_{i}\simeq\mathbb{P}^{1}$ in place of $q_{i}$. Locally, in $(u_{1},v_{1})$  (resp.~$(U_{1},V_{1})$)
chart $F_{i}$ is given by the equation $u_{i}=0$ (resp.~$V_{i}=0$) and is parameterized by $v_{i}$ (resp.~$U_{i}$) with $U_{i} = 1/v_{i}$. We then need to extend 
the mapping to $F_{i}$ via coordinate substitution and then check to see whether there are any new base points on $F_{i}$.

In particular, after the substitution 
$(x,y)=(u_{1},u_{1}v_{1})$ the expression for $\underline{x}$ becomes
\begin{equation*}
	\underline{x}(u_{1},v_{1})= \frac{u_1 v_1 \left(s-u_1 v_1\right)}{v_1 \left(u_1 \left(\alpha +2 n+v_1\right)-s\right)-n s},
\end{equation*}
and so we see a new base point $q_{2}(u_{1}=0,v_{1}=-n)$. 
Performing a similar computation in the $(U_{1},V_{1})$ chart and then in new $(u_{2},v_{2})$ and $(U_{2},V_{2})$ charts gives no new base points for 
the whole backward half-mapping $\tilde{\varphi}_{2}$ and the mapping 
in these charts is now well-defined everywhere. We call this configuration of base points 
\begin{equation*}
q_{1}(x=0,y=0)\leftarrow q_{2}(u_{1}=0,v_{1}=-n)	
\end{equation*}
a \emph{degeneration cascade}
resolving the base point $q_{1}$. We get a similar degeneration cascade resolving the point $q_{3}$
\begin{equation*}
	q_{3}(x=0,y=s)\leftarrow q_{4}(u_{3}=0,v_{3}=-n-\alpha).
\end{equation*}
There are no more singular points
for the backward half-mapping, but for the forward half-mapping $\tilde{\varphi}_{1}$, when considered in the $(X,Y)$ chart, we get another cascade of base points,
\begin{align*}
	q_5\left(X=0,Y=0\right)&\leftarrow 
	q_6\left(u_5=X=0,v_5=xY=0\right)\leftarrow 
	q_7\left(u_6=X,v_6=x^{2}Y=-1\right)\\
           &\leftarrow q_8(u_7=X,v_7=x(1 + x^{2}Y)=1+2n+\alpha).
\end{align*}

Blowing up all $8$ points $q_{i}$ gives us a (family of) rational algebraic surfaces parameterized by $n,\alpha,s$ (coordinates of the 
base points); $\mathcal{X} = \mathcal{X}_{n,\alpha,s}$. 
An important object associated with this family is the \emph{Picard lattice} 
\begin{equation*}
	\operatorname{Pic}(\mathcal{X}) = \operatorname{Span}_{\mathbb{Z}}\{\mathcal{H}_{x}, \mathcal{H}_{y},\mathcal{F}_{1},\ldots \mathcal{F}_{8}\}\simeq 
	\operatorname{Cl}(\mathcal{X})
\end{equation*}
generated by the classes $\mathcal{H}_{x,y}=[H_{x,y=k}]$ of coordinate lines and classes $\mathcal{F}_{i} = [F_{i}]$ of the exceptional divisors. 
$\operatorname{Pic}(\mathcal{X})$ is equipped with the symmetric bilinear \emph{intersection form} given by 
\begin{equation}\label{eq:int-form}
\mathcal{H}_{x}\bullet \mathcal{H}_{x} = \mathcal{H}_{y}\bullet \mathcal{H}_{y} = \mathcal{H}_{x}\bullet \mathcal{F}_{i} = 
\mathcal{H}_{y}\bullet \mathcal{F}_{j} = 0,\qquad \mathcal{H}_{x}\bullet \mathcal{H}_{y} = 1,\qquad  \mathcal{F}_{i}\bullet \mathcal{F}_{j} = - \delta_{ij}	
\end{equation}
on the generators, and then extended by linearity.

This configuration of the base points and the rational algebraic surface $\mathcal{X}$ obtained after we do all eight blowups
are shown on Figure~\ref{fig:Laguerre-pt-conf}. We also mark on the surface $\mathcal{X}$ a number of special curves, such as the exceptional divisors 
$F_{i}$ and some other curves with the self-intersection $-1$ and $-2$.

\begin{figure}[ht]
	\begin{center}		
	\begin{tikzpicture}[>=stealth,basept/.style={circle, draw=red!100, fill=red!100, thick, inner sep=0pt,minimum size=1.2mm}]
	\begin{scope}[xshift=0cm,yshift=0cm]
	\draw [black, line width = 1pt] (-0.2,0) -- (3.2,0)	node [pos=0,left] {\small $H_{y}$} node [pos=1,right] {\small $y=0$};
	\draw [black, line width = 1pt] (-0.2,3) -- (3.2,3) node [pos=0,left] {\small $H_{y}$} node [pos=1,right] {\small $y=\infty$};
	\draw [black, line width = 1pt] (0,-0.2) -- (0,3.2) node [pos=0,below] {\small $H_{x}$} node [pos=1,xshift = -7pt, yshift=5pt] {\small $x=0$};
	\draw [black, line width = 1pt] (3,-0.2) -- (3,3.2) node [pos=0,below] {\small $H_{x}$} node [pos=1,xshift = 7pt, yshift=5pt] {\small $x=\infty$};
	\node (q1) at (0,0) [basept,label={[xshift = 10pt, yshift=-3pt] \small $q_{1}$}] {};
	\node (q2) at (-0.5,0.5) [basept,label={[xshift = -7pt, yshift=-7pt] \small $q_{2}$}] {};
	\node (q3) at (0,1) [basept,label={[xshift = 10pt, yshift=-3pt] \small $q_{3}$}] {};
	\node (q4) at (-0.5,1.5) [basept,label={[xshift = -7pt, yshift=-7pt] \small $q_{4}$}] {};
	\node (q5) at (3,3) [basept,label={[xshift = -7pt, yshift=-15pt] \small $q_{5}$}] {};
	\node (q6) at (2.3,3) [basept,label={[xshift=3pt,yshift=0pt] \small $q_{6}$}] {};
	\node (q7) at (1.8,3.5) [basept,label={[xshift=3pt,yshift=-3pt] \small $q_{7}$}] {};
	\node (q8) at (1.3,3.5) [basept,label={[yshift=-3pt] \small $q_{8}$}] {};
	\draw [red, line width = 0.8pt, ->] (q2) -- (q1);
	\draw [red, line width = 0.8pt, ->] (q4) -- (q3);
	\draw [red, line width = 0.8pt, ->] (q6) -- (q5);
	\draw [red, line width = 0.8pt, ->] (q7) -- (q6);
	\draw [red, line width = 0.8pt, ->] (q8) -- (q7);
	\end{scope}
	\draw [->] (6,1.5)--(4,1.5) node[pos=0.5, below] {$\operatorname{Bl}_{q_{1}\cdots q_{8}}$};
	\begin{scope}[xshift=8.5cm,yshift=0cm]
	\draw[red, line width = 1pt] (0.3,0) -- (3.7,0) node [pos=0.5,above] {\small $H_{y}-F_{1}$};
	\draw[red, line width = 1pt] (3.2,-0.2) -- (4.2,0.8) node [pos=0,below] {\small $H_{x}-F_{5}$};
	\draw[blue, line width = 1pt] (-0.2,0.8) -- (0.8,-0.2) node [pos=1,below] {\small $F_{1}-F_{2}$};
	\draw[red, line width = 1pt] (0.1,0.1) -- (1.1,1.1) node [pos=1,right] {\small $F_{2}$};
	\draw[blue, line width = 1pt] (0,0.3) -- (0,3.2) node [pos=1,above] {\small $H_{x}-F_{13}$};
	\draw[blue, line width = 1pt] (-0.8,1.5) -- (0.5,1.5) node [pos=0,left] {\small $F_{3}-F_{4}$};
	\draw[red, line width = 1pt] (-0.5,1) -- (-0.5,2) node [pos=1,above] {\small $F_{4}$};
	\draw[blue, line width = 1pt] (4,0.3) -- (4,2.7) node [pos=0.5,right] {\small $F_{5}-F_{6}$};
	\draw[blue, line width = 1pt] (4.2,2.2) -- (3.2,3.2) node [pos=1,above] {\small $F_{6}-F_{7}$};
	\draw[blue, line width = 1pt] (-0.2,3) -- (3.7,3) node [pos=0,left] {\small $H_{y}-F_{56}$};
	\draw[blue, line width = 1pt] (3.9,2.9) -- (2.9,1.9) node [pos=1,left] {\small $F_{7}-F_{8}$};
	\draw[red, line width = 1pt] (2.8,2.8) -- (3.8,1.8) node [pos=0, yshift=-3pt, left] {\small $F_{8}$};
	\end{scope}
	\end{tikzpicture}
	\end{center}
	\caption{The Sakai Surface for the generalized Laguerre recurrence} 
	\label{fig:Laguerre-pt-conf}
\end{figure}	

Note that this configuration of the blowup points lies on a (reducible) bi-quadratic curve given by the equation 
$x^{2}=0$ in the affine $(x,y)$-chart. This curve is the polar divisor of a symplectic form $\omega = (k/x^{2}) dx \wedge dy$,
where $k$ is an arbitrary constant,
and the class of the divisor of this symplectic form $[\omega]$ in the \emph{Picard lattice} is called the \emph{canonical} 
divisor class, $\mathcal{K}_{\mathbb{P}^{1}\times \mathbb{P}^{1}}$. 
We are interested in the dual, or \emph{anti-canonical}, divisor class of the lift of $\omega$ to the surface $\mathcal{X}$:
\begin{equation}
	-\mathcal{K}_{\mathcal{X}} = -[K_\mathcal{X}] = 2 \mathcal{H}_{x} + 2 \mathcal{H}_{y} - \mathcal{F}_{1} - \cdots - \mathcal{F}_{8}
	= \delta_{0} + \delta_{1} + 2 \delta_{2} + 2 \delta_{3} + 2 \delta_{4} + \delta_{5} + \delta_{6},
\end{equation}
where $\delta_{i} = [d_{i}]$ are classes of the \emph{irreducible components} of the anti-canonical divisor $K_{\mathcal{X}}$
shown on Figure~\ref{fig:Laguerre-surface-roots}. The irreducible components $d_{i}$ are $-2$ curves that are shown on Figure~\ref{fig:Laguerre-pt-conf}
on the right. Note that sometimes it is convenient to abbreviate 
$\mathcal{F}_{i_{1}}+\mathcal{F}_{i_{2}}+\cdots +\mathcal{F}_{i_{k}}$ as $\mathcal{F}_{i_{1}i_{2}\cdots i_{k}}$.
\begin{figure}[ht]
	\begin{equation}\label{eq:Laguerre-surface-rb-prelim}
		\raisebox{-32.1pt}{\begin{tikzpicture}[
			elt/.style={circle,draw=black!100,thick, inner sep=0pt,minimum size=2mm}]
			\path
			(-1,1)
			node (d0) [elt, label={[xshift=-10pt, yshift = -10 pt] $\delta_{0}$} ] {}
			(-1,-1)
			node (d1) [elt, label={[xshift=-10pt, yshift = -10 pt] $\delta_{1}$} ] {}
			( 0,0)
			node   (d2) [elt, label={[xshift=0pt, yshift = -25 pt] $\delta_{2}$} ] {}
			( 1,0)
			node   (d3) [elt, label={[xshift=0pt, yshift = -25 pt] $\delta_{3}$} ] {}
			( 2,0)
			node   (d4) [elt, label={[xshift=0pt, yshift = -25 pt] $\delta_{4}$} ] {}
			( 3,1)
			node   (d6) [elt, label={[xshift=10pt, yshift = -10 pt] $\delta_{6}$} ] {}
			( 3,-1) node (d5) [elt, label={[xshift=10pt, yshift = -10 pt] $\delta_{5}$} ] {};
			\draw [black,line width=1pt ] (d0) -- (d2) -- (d1)  (d2) -- (d3) --(d4) (d6) -- (d4) -- (d5);
			\end{tikzpicture}} \qquad
		\begin{alignedat}{2}
			\delta_{0} &=  \mathcal{F}_{1}-\mathcal{F}_{2}, &\qquad  \delta_{4} &=\mathcal{F}_{6}-\mathcal{F}_{7},\\
			\delta_{1} &= \mathcal{F}_{3}-\mathcal{F}_{4}, &\qquad  \delta_{5} &= \mathcal{F}_{5}-\mathcal{F}_{6},\\
			\delta_{2} &= \mathcal{H}_{x} - \mathcal{F}_{13}, &\qquad  \delta_{6} &= \mathcal{F}_{7}-\mathcal{F}_{8}.\\
			\delta_{3} &= \mathcal{H}_{y} - \mathcal{F}_{56}, &\qquad
		\end{alignedat}
	\end{equation}
	\caption{The (preliminary) Surface Root Basis for the generalized Laguerre recurrence}
	\label{fig:Laguerre-surface-roots}
\end{figure}
Elements $\delta_{i}\in \operatorname{Pic}(\mathcal{X})$ constitute the \emph{surface root basis} of the surface family $\mathcal{X}_{n,\alpha,s}$ and their
intersection configuration is encoded in an affine Dynkin diagram shown on Figure~\ref{fig:Laguerre-surface-roots}. This diagram has the type
$D_{6}^{(1)}$ and so we expect that our recurrence is given by a discrete Painlev\'e equation in this family. In 
Appendix~\ref{sec:dP-2A1-std} we collect some essential algebro-geometric facts about discrete Painlev\'e equations of surface type $D_{6}^{(1)}$.
It now remains to determine which  discrete Painlev\'e equation our dynamics corresponds to. 
For that, we need to look at the induced mapping on the 
Picard lattice, and associated action on the surface and symmetry root bases, we do that next. 

\subsection{Induced linear mappings on the Picard lattice} 
\label{sub:Laguerre-Picard}
We now compute the linearized action on the Picard lattice for the half-maps $\tilde{\varphi}_{1}$, $\tilde{\varphi}_{2}$, as well as the full \emph{forward map}
$\tilde{\varphi}_{n} = (\tilde{\varphi}_{2}^{n+1})^{-1}\circ\tilde{\varphi}_{1}^{n}: (x_{n},y_{n})\to (x_{n+1},y_{n+1})$ (that we 
write simply as $\tilde{\varphi}:(x,y)\to (\overline{x},\overline{y})$) and the full \emph{backward map} 
$\tilde{\varphi}^{-1}:(x,y)\to (\underline{x},\underline{y})$. This computation is standard and is explained in detail, for example, in \cite{DzhFilSto:2020:RCDOPWHWDPE,DzhTak:2018:SASGTDPE},
so here we only state the result. 

\begin{lemma}\label{lem:Laguerre-Pic-action}
	\qquad 
	
	\begin{enumerate}[(a)]
		\item The action of the half-step forward dynamic $(\tilde{\varphi}_{1})_{*}:\operatorname{Pic}(\mathcal{X}) \to \operatorname{Pic}(\overline{\mathcal{X}})$
		is given by 
		\begin{align*}
			\mathcal{H}_{x}&\mapsto \overline{\mathcal{H}}_{x}, &\quad 
			\mathcal{F}_{4}&\mapsto \overline{\mathcal{F}}_{2},
			\\
			\mathcal{H}_{y}&\mapsto 2\overline{\mathcal{H}}_{x} + \overline{\mathcal{H}}_{y} - \overline{\mathcal{F}}_{5678},&\quad 			 
			\mathcal{F}_{5}&\mapsto \overline{\mathcal{H}}_{x} - \overline{\mathcal{F}}_{8},
			\\
			\mathcal{F}_{1}&\mapsto \overline{\mathcal{F}}_{3},&\quad 			 
			\mathcal{F}_{6}&\mapsto \overline{\mathcal{H}}_{x} - \overline{\mathcal{F}}_{7},
			\\
			\mathcal{F}_{2}&\mapsto \overline{\mathcal{F}}_{4},&\quad 			 
			\mathcal{F}_{7}&\mapsto \overline{\mathcal{H}}_{x} - \overline{\mathcal{F}}_{6},
			\\
			\mathcal{F}_{3}&\mapsto \overline{\mathcal{F}}_{1},&\quad 			 
			\mathcal{F}_{8}&\mapsto \overline{\mathcal{H}}_{x} - \overline{\mathcal{F}}_{5}.
		\end{align*}
		Under this mapping points $q_{2}$ and $q_{4}$ evolve as  $n\mapsto n+1$ and as a result 
		we get $\overline{q}_{2}(0,-n-1)$, $\overline{q}_{4}(0,-n-1-\alpha)$, but 
		all of the other points do not evolve.

		\item The action of the half-step backward dynamic $(\tilde{\varphi}_{2})_{*}:\operatorname{Pic}(\mathcal{X}) \to \operatorname{Pic}(\underline{\mathcal{X}})$
		is given by 
		\begin{align*}
			\mathcal{H}_{x}&\mapsto \underline{\mathcal{H}}_{x} + 2 \underline{\mathcal{H}}_{y} - \underline{\mathcal{F}}_{1234}, &\quad 
			\mathcal{F}_{4}&\mapsto \underline{\mathcal{H}}_{y} - \underline{\mathcal{F}}_{3},
			\\
			\mathcal{H}_{y}&\mapsto \underline{\mathcal{H}}_{y},&\quad 			 
			\mathcal{F}_{5}&\mapsto \underline{\mathcal{F}}_{5},
			\\
			\mathcal{F}_{1}&\mapsto \underline{\mathcal{H}}_{y} - \underline{\mathcal{F}}_{2},&\quad 			 
			\mathcal{F}_{6}&\mapsto \underline{\mathcal{F}}_{6},
			\\
			\mathcal{F}_{2}&\mapsto \underline{\mathcal{H}}_{y} - \underline{\mathcal{F}}_{1},&\quad 			 
			\mathcal{F}_{7}&\mapsto \underline{\mathcal{F}}_{7},
			\\
			\mathcal{F}_{3}&\mapsto \underline{\mathcal{H}}_{y} - \underline{\mathcal{F}}_{4},&\quad 			 
			\mathcal{F}_{8}&\mapsto \underline{\mathcal{F}}_{8}.
		\end{align*}
		Under this mapping the point $q_{8}$ evolves as  $n\mapsto n-1$ and we get $\underline{q}_{8}(0,-1 + 2n + \alpha)$. All the other points do not evolve.

		\item The action of the full forward dynamic $(\tilde{\varphi})_{*}:\operatorname{Pic}(\mathcal{X}) \to \operatorname{Pic}(\overline{\mathcal{X}})$
		is given by 
		\begin{align*}
			\mathcal{H}_{x}&\mapsto \overline{\mathcal{H}}_{x} + 2\overline{\mathcal{H}}_{y} - \overline{\mathcal{F}}_{1234}, &\quad 
			\mathcal{F}_{4}&\mapsto \overline{\mathcal{H}}_{y} - \overline{\mathcal{F}}_{1},
			\\
			\mathcal{H}_{y}&\mapsto 2\overline{\mathcal{H}}_{x} + 5 \overline{\mathcal{H}}_{y} - 2\overline{\mathcal{F}}_{1234} - \overline{\mathcal{F}}_{5678},&\quad 			 
			\mathcal{F}_{5}&\mapsto \overline{\mathcal{H}}_{x} + 2\overline{\mathcal{H}}_{y} - \overline{\mathcal{F}}_{12348},
			\\
			\mathcal{F}_{1}&\mapsto \overline{\mathcal{H}}_{y} - \overline{\mathcal{F}}_{4},&\quad 			 
			\mathcal{F}_{6}&\mapsto \overline{\mathcal{H}}_{x} + 2\overline{\mathcal{H}}_{y} - \overline{\mathcal{F}}_{12347},
			\\
			\mathcal{F}_{2}&\mapsto \overline{\mathcal{H}}_{y} - \overline{\mathcal{F}}_{3},&\quad 			 
			\mathcal{F}_{7}&\mapsto \overline{\mathcal{H}}_{x} + 2\overline{\mathcal{H}}_{y} - \overline{\mathcal{F}}_{12346},
			\\
			\mathcal{F}_{3}&\mapsto \overline{\mathcal{H}}_{y} - \overline{\mathcal{F}}_{2},&\quad 			 
			\mathcal{F}_{8}&\mapsto \overline{\mathcal{H}}_{x} + 2\overline{\mathcal{H}}_{y} - \overline{\mathcal{F}}_{12345}.
		\end{align*}
		Under this mapping we get the expected evolution of the discrete parameter $n\mapsto n+1$, so we get $\overline{q}_{2}(0,-n-1)$, 
		$\overline{q}_{4}(0,-n-1-\alpha)$, and $\overline{q}_{8}(0,3 + 2n + \alpha)$.

		\item The action of the full backward dynamic 
		$(\tilde{\varphi})^{-1}_{*}:\operatorname{Pic}(\mathcal{X}) \to \operatorname{Pic}(\underline{\mathcal{X}})$ is given by 
		\begin{align*}
			\mathcal{H}_{x}&\mapsto 5\underline{\mathcal{H}}_{x} + 2\underline{\mathcal{H}}_{y} - \underline{\mathcal{F}}_{1234} 
			-2\underline{\mathcal{F}}_{5678}, &\quad 
			\mathcal{F}_{4}&\mapsto 2\underline{\mathcal{H}}_{x} + \underline{\mathcal{H}}_{y} - \underline{\mathcal{F}}_{15678},
			\\
			\mathcal{H}_{y}&\mapsto 2\underline{\mathcal{H}}_{x} + \underline{\mathcal{H}}_{y} - \underline{\mathcal{F}}_{5678},&\quad 			 
			\mathcal{F}_{5}&\mapsto \underline{\mathcal{H}}_{x}  - \underline{\mathcal{F}}_{8},
			\\
			\mathcal{F}_{1}&\mapsto 2\underline{\mathcal{H}}_{x} + \underline{\mathcal{H}}_{y} - \underline{\mathcal{F}}_{45678},&\quad 			 
			\mathcal{F}_{6}&\mapsto \underline{\mathcal{H}}_{x}  - \underline{\mathcal{F}}_{7},
			\\
			\mathcal{F}_{2}&\mapsto 2\underline{\mathcal{H}}_{x} + \underline{\mathcal{H}}_{y} - \underline{\mathcal{F}}_{35678},&\quad 			 
			\mathcal{F}_{7}&\mapsto \underline{\mathcal{H}}_{x}  - \underline{\mathcal{F}}_{6},
			\\
			\mathcal{F}_{3}&\mapsto 2\underline{\mathcal{H}}_{x} + \underline{\mathcal{H}}_{y} - \underline{\mathcal{F}}_{25678},&\quad 			 
			\mathcal{F}_{8}&\mapsto \underline{\mathcal{H}}_{x}  - \underline{\mathcal{F}}_{5}.
		\end{align*}
		Under this mapping we get the expected evolution of the discrete parameter $n\mapsto n-1$, so we get $\underline{q}_{2}(0,-n+1)$, 
		$\underline{q}_{4}(0,-n+1-\alpha)$, and $\underline{q}_{8}(0,-1 + 2n + \alpha)$.

	\end{enumerate}
\end{lemma}

\begin{remark} It may seem strange that for half-maps in (a) and (b) some base points evolve as  $n$ evolves, but some do not, this does not happen for full maps. 
	The reason for that is that for half-maps the evolution on the Picard lattice is not a translation, and so root variables (and hence the parameters of the base
	points) evolve in a more complicated fashion. We briefly explain that in Section~\ref{sub:Laguerre-half-maps}.	
\end{remark}

\subsection{Identifying the dynamics on the level of the Picard lattice} 
\label{sub:Laguerre-Pic-match}
We are now in the position to start identifying the dynamics. For that, we first need to find some change of basis of the Picard lattice that matches 
the corresponding surface and symmetry root bases. As usual, we begin with the preliminary change of basis that identifies the symmetry roots.
It is given in the next Lemma.

\begin{lemma}\label{lem:Laguerre2KNY-prelim-pic}
	The change of basis that identifies the surface roots corresponding to the irreducible components of the anti-canonical divisors for the standard 
	surface on Figure~\ref{fig:KNY-pt-conf} and the Sakai surface  obtained by regularizing the generalized Laguerre recurrence 
	on Figure~\ref{fig:Laguerre-pt-conf} is  given by
 $$\mathcal{H}_{x} = \mathcal{H}_{q}, \quad \mathcal{H}_{y} = \mathcal{H}_{p}, \quad
	\mathcal{F}_{i} = \mathcal{E}_{i}, \quad i=1\cdots 8.$$
\end{lemma}

\begin{proof} This change of basis follows immediately when we compare the root bases on Figures~\ref{fig:Laguerre-surface-roots} and 
	\ref{fig:KNY-surface-roots}.
\end{proof}

With this identification and root bases depicted in Figure~\ref{fig:standard-symmetry-roots}, we get the  preliminary symmetry root basis for the Laguerre recurrence to be 
	\begin{equation*}
		\alpha_{0} = 2\mathcal{H}_x+\mathcal{H}_y-\mathcal{F}_{345678},\quad \alpha_{1}= \mathcal{H}_y-\mathcal{F}_{12},\quad
		\alpha_{2} = \mathcal{H}_y-\mathcal{F}_{34}, \quad 
		\alpha_{3}= 2\mathcal{H}_x+\mathcal{H}_y-\mathcal{F}_{125678},			
	\end{equation*}
	and the full forward evolution $(\tilde{\varphi})_{*}$ from Lemma~\ref{lem:Laguerre-Pic-action}(c) acts on the symmetry roots as 
	\begin{equation*}
	\tilde{\varphi}_*:\upalpha=\langle\alpha_0,\alpha_1,\alpha_2,\alpha_3\rangle\mapsto
	\tilde{\varphi}_*(\upalpha)=\upalpha+ \langle -1,1,1,-1 \rangle\delta,
	\end{equation*}
	which is the \emph{opposite} of the standard translation \eqref{eq:dP-KNY-trans}. To \emph{reverse} the translation direction we can act 
	on the symmetry roots by the Dynkin diagram automorphisms $\sigma_{i}$ described in Section~\ref{sub:the_extended_affine_weyl_symmetry_group} as 
	follows:
	\begin{equation*}
		\sigma_{2}\circ \sigma_{1}\circ \sigma_{2}\circ \sigma_{1} = 
		(\alpha_{0}\alpha_{1})(\alpha_{0}\alpha_{3})(\alpha_{1}\alpha_{2})(\alpha_{0}\alpha_{1})(\alpha_{0}\alpha_{3})(\alpha_{1}\alpha_{2})
		= (\alpha_{0}\alpha_{1})(\alpha_{2}\alpha_{3}).
	\end{equation*} 

This action then adjusts our change of basis and this results in the correct change of basis given in the following Lemma.
\begin{lemma}\label{lem:Laguerre2KNY-pic}
	The change of basis of the Picard lattice identifying both the geometry and the standard dynamics between the Laguerre recursion and the 
	standard discrete Painlev\'e dynamics \eqref{eq:dP-KNY-map} is given by:

		\begin{equation}\label{eq:Laguerre-Sakai-Pic}
			\begin{aligned}
					\mathcal{H}_{x} &=\mathcal{H}_{q}, &\qquad
					\mathcal{H}_{q} &= \mathcal{H}_{x}, \\
					\mathcal{H}_{y} &= 2\mathcal{H}_{q} + \mathcal{H}_{p} - \mathcal{E}_{5678}, &\qquad
					\mathcal{H}_{p} &= 2\mathcal{H}_{x} + \mathcal{H}_{y} - \mathcal{F}_{5678}, \\
					\mathcal{F}_{1} &= \mathcal{E}_{3}, &\qquad
					\mathcal{E}_{1} &= \mathcal{F}_{3},\\
					\mathcal{F}_{2} &= \mathcal{E}_{4}, &\qquad
					\mathcal{E}_{2} &= \mathcal{F}_{4},\\
					\mathcal{F}_{3} &= \mathcal{E}_{1}, &\qquad
					\mathcal{E}_{3} &= \mathcal{F}_{1},\\
					\mathcal{F}_{4} &= \mathcal{E}_{2}, &\qquad
					\mathcal{E}_{4} &= \mathcal{F}_{2},\\
					\mathcal{F}_{5} &= \mathcal{H}_{q} - \mathcal{E}_{8}, &\qquad
					\mathcal{E}_{5} &= \mathcal{H}_{x} - \mathcal{F}_{8},\\
					\mathcal{F}_{6} &= \mathcal{H}_{q} - \mathcal{E}_{7}, &\qquad
					\mathcal{E}_{6} &= \mathcal{H}_{x} - \mathcal{F}_{7},\\
					\mathcal{F}_{7} &= \mathcal{H}_{q} - \mathcal{E}_{6}, &\qquad
					\mathcal{E}_{7} &= \mathcal{H}_{x} - \mathcal{F}_{6},\\
					\mathcal{F}_{8} &= \mathcal{H}_{q} - \mathcal{E}_{5}, &\qquad
					\mathcal{E}_{8} &= \mathcal{H}_{x} - \mathcal{F}_{5}.
			\end{aligned}
		\end{equation}

\end{lemma}

	As expected, this change of bases results in switching the symmetry roots and the new symmetry root basis for the Laguerre recurrence is 
	\begin{equation}\label{eq:Laguerre-symmetry-rb}
		\begin{aligned}
			\alpha_{0} &= \mathcal{H}_{y}- \mathcal{F}_{12}, &\quad 
			\alpha_{2} &= 2\mathcal{H}_{x} + \mathcal{H}_{y} - \mathcal{F}_{125678},\\ 
			\alpha_{1} &= 2\mathcal{H}_{x} + \mathcal{H}_{y} - \mathcal{F}_{345678}, &\quad 
			\alpha_{3} &= \mathcal{H}_{y}- \mathcal{F}_{34},
		\end{aligned}
	\end{equation}
	and on this root basis the translation dynamics becomes standard, 
	\begin{equation}\label{eq:Laguerre-translation}
	\tilde{\varphi}_*(\upalpha)=\upalpha+ \langle 1,-1,-1,1 \rangle\delta.	
	\end{equation}
	
	Note that this change of basis also acts on the surface root basis by the Dynkin diagram automorphisms, so our surface root basis
	permutes as
	\begin{equation}\label{eq:Laguerre-surface-rb}
		\begin{aligned}
			\delta_{0} &=  \mathcal{F}_{3}-\mathcal{F}_{4}, &\quad  
			\delta_{2} &=  \mathcal{H}_{x} - \mathcal{F}_{13}, &\quad
			\delta_{4} &=  \mathcal{F}_{6}-\mathcal{F}_{7}, &\quad
			\delta_{6} &=  \mathcal{F}_{5}-\mathcal{F}_{6}.
			\\
			\delta_{1} &= \mathcal{F}_{1}-\mathcal{F}_{2}, &\quad  
			\delta_{3} &= \mathcal{H}_{y} - \mathcal{F}_{56}, &\quad  
			\delta_{5} &= \mathcal{F}_{7}-\mathcal{F}_{8}, &\qquad  
		\end{aligned}
	\end{equation}


\subsection{Root variables and  the parameter matching} 
\label{sub:Laguerre-KNY-pars}
	Before finding the change of coordinates that identifies the two dynamics, we need to match the Laguerre weight parameters $\alpha$, $s$ and the 
	recurrence step $n$ with the standard parameters, or the \emph{root variables} $a_{i}$.  Recall that the root variables are given 
	by the \emph{Period Map} $\chi: Q=\operatorname{Span}_{\mathbb{Z}}{\alpha_i} \rightarrow \mathbb{C}$, $a_{i} = \chi(\alpha_{i})$. 
	To define the
	period map we need to find the symplectic form $\omega$ such that $-[\omega] = -\mathcal{K}_{\mathcal{X}}$, i.e., so that the polar divisor 
	of $\omega$ is a $(2,2)$-curve on $\mathbb{P}^{1} \times \mathbb{P}^{1}$ passing through all the base points $q_{1},\ldots, q_{8}$,
	which, in our case, is $\omega = (k/x^{2}) dx \wedge dy$, when written in the  affine chart $(x,y)$. 
	The period map computation is then based on the following facts \cite{Sak:2001:RSAWARSGPE}:
	
	\begin{itemize}
		\item each symmetry root $\alpha_{i}$ can be represented (in a non-unique way) as a difference of classes of two effective divisors, 
		$\alpha_i=[C_i^1]-[C_i^0]$;
		\item for each such representation there exists a unique irreducible component $d_k$ of the anti-canonical divisor $-K_{\mathcal{X}}$ 
		such that $d_k \bullet C_i^1=d_k \bullet C_i^0=1$; 
		put $P_i=d_k\cap C_i^0$ and $Q_i=d_k\cap C_i^1$.
	\end{itemize}
	Then 
	\begin{equation}\label{eq:period-map}
		\chi(\alpha_i)=\chi([C_i^1]-[C_i^0])=\int_{P_i}^{Q_i}\frac{1}{2\pi i}\oint_{d_k}\omega=\int_{P_i}^{Q_i}\operatorname{res}_{d_k}\omega.
	\end{equation}
	In our case, the period map is given by the following Lemma. 
	
	\begin{lemma}\label{lem:period_map-GenLaguerre} 
		\qquad
		\begin{enumerate}[(i)]
			\item The residues of the symplectic form $\omega$ along the irreducible components $d_{i}$ of the polar divisor corresponding to the 
			surface roots $\delta_{i} = [d_{i}]$, given in \eqref{eq:Laguerre-surface-rb},  are given by
			\begin{align*}
				\operatorname{res}_{d_{0}} \omega &= k\, dv_{3}, &\qquad 
					 \operatorname{res}_{d_{2}} \omega &=0, &\qquad   \operatorname{res}_{d_{4}} \omega &= 0 ,&\qquad  
				\operatorname{res}_{d_{6}} \omega &= \frac{k}{v_{5}^{2}}\, dv_{5}.\\
				\operatorname{res}_{d_{1}} \omega &= k\, dv_{1}, &\qquad  \operatorname{res}_{d_{3}} \omega &= 0,&\qquad   
				\operatorname{res}_{d_{5}} \omega &=  k\, dv_{7},
			\end{align*}
			\item For the standard root variable normalization $\chi(\delta) = a_{0}+a_{1} = a_{1}+a_{2} =1$ we need to take 
			$k=-1$ and the root variables $a_{i}$ are then given by 
			\begin{equation}\label{eq:Laguerre-root-vars}
				a_{0} = -n,\quad a_{1} = n+1,\quad a_{2} = 1 + n + \alpha,\quad a_{3} = -(n+\alpha).
			\end{equation}			
		\end{enumerate}
	\end{lemma}

	\begin{proof} Detailed examples of such computations can be found in \cite{DzhTak:2018:SASGTDPE,DzhFilSto:2020:RCDOPWHWDPE}, so here we only explain 
		one such computation. Consider the root $\alpha_{1}$ and represent it as a difference of two effective classes:
		\begin{equation*}
			\alpha_{1} = 2\mathcal{H}_{x} + \mathcal{H}_{y} - \mathcal{F}_{3} - \mathcal{F}_{4}- \mathcal{F}_{5}- \mathcal{F}_{6}- \mathcal{F}_{7} - \mathcal{F}_{8}
			= [2H_{x} + H_{y} - F_{3} -F_{4} -F_{5} -F_{6}-F_{7}] - [F_{8}].
		\end{equation*}
		The first class is a class of a proper transform of a $(2,1)$-curve passing through the points $q_{3},\ldots, q_{7}$ and a 
		direct computation shows that its equation in the affine 
		$(x,y)$-chart is $x(x+n+\alpha)+y-s=0$. The second class is the class of the exceptional divisor $F_{8}$, so we need to consider these curves 
		in the $(u_{7},v_{7})$-chart. In this chart the proper transform $2H_{x} + H_{y} - F_{34567}$ is given by the equation 
		$v_{7}-(n+\alpha) +u_{7}(s + (n+\alpha)v_{7} - s u_{7}v_{7}) = 0$, which intersects with the irreducible divisor $d_{5} = F_{7}-F_{8}$,
		given by the equation $u_{7} = 0$, at the point $(u_{7},v_{7}) = (0,n+\alpha)$. The exceptional divisor $F_{8}$ intersects with $d_{5}$ at the point 
		$q_{8}(0,1 + 2n + \alpha)$. Computing the symplectic form $\omega$ in the $(u_{7},v_{7})$-chart,
		\begin{equation*}
			\omega = k \frac{dx \wedge dy}{x^{2}} = k \frac{dX \wedge dY}{Y^{2}} =  k \frac{du_{5} \wedge dv_{5}}{u_{5} v_{5}^{2}} = 
			k \frac{du_{6} \wedge dv_{6}}{u_{6}^{2} v_{6}^{2}} = k \frac{du_{7} \wedge dv_{7}}{u_{7} (1-u_{7}v_{7})^{2}},
		\end{equation*}
		we see that 
		\begin{equation*}
			\operatorname{res}_{d_{5}} \omega =  \operatorname{res}_{u_{7}=0} k \frac{du_{7} \wedge dv_{7}}{u_{7} (1-u_{7}v_{7})^{2}} =  k \, dv_{7},
			\qquad \chi(\alpha_{1}) = \int_{v_{7} = 1 + 2n + \alpha}^{v_{7} = n + \alpha} k \, dv_{7} = -k(n+1).
		\end{equation*}
		Other computations are similar and we get 
		\begin{equation*}
			a_{0} = kn,\quad a_{1} =-k (n+1),\quad a_{2} = -k(1 + n + \alpha),\quad a_{3} = k(n+\alpha).
		\end{equation*}
		The normalization $\chi(\delta) = a_{0}+a_{1} = a_{1}+a_{2} = - k = 1$ then tells us that $k=-1$.
	\end{proof}
	
	\begin{remark} Note that the root variable evolution under the discrete step $n\mapsto n+1$ is given by 
		\begin{equation*}
		\overline{a}_{0} = a_{0} - 1,\quad \overline{a}_{1} = a_{1} + 1,\quad \overline{a}_{2} = a_{2} + 1,\quad \overline{a}_{3} = a_{3}-1,	
		\end{equation*}
		which corresponds to the correct (dual) translation on the root basis given by \eqref{eq:Laguerre-translation}.		
	\end{remark}

\subsection{The change of coordinates} 
\label{sub:Laguerre-to-KNY}
We are finally in the position to prove Theorem~\ref{thm:Laguerre-coord-change} that gives the explicit change of coordinates matching 
the dynamics \eqref{eq:Laguerre-xy-rec} and \eqref{eq:dP-KNY-map}. For that, we use the change of basis \eqref{eq:Laguerre-Sakai-Pic}
in Lemma~\ref{lem:Laguerre2KNY-pic}. This computation is standard, see \cite{DzhTak:2018:SASGTDPE,DzhFilSto:2020:RCDOPWHWDPE} for 
detailed examples, so here we only outline it. From the change of basis on the Picard lattice for the coordinate classes, 
\begin{equation*}
	\mathcal{H}_{q} = \mathcal{H}_{x}, \qquad
	\mathcal{H}_{p} = 2\mathcal{H}_{x} + \mathcal{H}_{y} - \mathcal{F}_{5} - \mathcal{F}_{6}- \mathcal{F}_{7} - \mathcal{F}_{8}, 	
\end{equation*}
we see that, up to M\"obius transformations, $q\sim x$, and $p$ is a (projective) coordinate on a pencil 
$|H_{p}| = |2H_{x} + H_{y} - F_{5678}|$ of $(2,1)$-curves passing through the points $q_{5},q_{6},q_{7},q_{8}$, and so 
$p\sim x^{2} + x(1 + 2n + \alpha) + y$. Thus, we take the change of coordinates to be 
\begin{equation*}
	q(x,y) = \frac{A x + B}{C x + D},\qquad p(x,y) = \frac{K(x^{2} + x(1 + 2n + \alpha) + y) + L}{M(x^{2} + x(1 + 2n + \alpha) + y) + N},
\end{equation*}
where $A,\ldots,N$ are parameters of M\"obius transformations that can be found from the knowledge of the divisor mapping, 
see \cite{CarDzhTak:2017:FDIMDPE}. For example, from the correspondence $F_{1}-F_{2} = E_{3} - E_{4}$ we see that under the change of coordinates
the point $q_{1}(0,0)$ should map to the point $p_{3}(\infty,0)$, i.e.,
\begin{equation*}
	Q(0,0) = \frac{C\cdot 0 + D}{A\cdot 0 + B} = \frac{D}{B} = 0,\quad p(0,0) = \frac{L}{N} = 0\implies D = L = 0,
\end{equation*}
and then we can take $C=K=1$ to get 
\begin{equation*}
	q(x,y) = A + \frac{B}{x},\qquad p(x,y) = \frac{x^{2} + x(1 + 2n + \alpha) + y}{M(x^{2} + x(1 + 2n + \alpha) + y) + N}.
\end{equation*}
From the correspondence $F_{7}-F_{8}=E_{5}-E_{6}$ we see that the change of coordinates mapping, when written in the charts $(u_{7},v_{7})\to (q,P)$
and restricted to $u_{7}=0$, should map to $(q=0, P=0)$:
\begin{align*}
	q(u_{7},v_{7}) &= A + B u_{7},\\
	P(u_{7},v_{7}) &= - \frac{u_{7} \left(v_{7} \left(M (\alpha +2 n+1)+N u_{7}\right)-N\right)-M \left(\alpha +2 n-v_{7}+1\right)
	}{\alpha -v_7 \left(u_7 (\alpha +2 n+1)+1\right)+2 n+1},
\end{align*}
so $q(0,v_{7})=A=0$ and $P(0,v_{7}) = M=0$. Thus, 
\begin{equation*}
	q(x,y) = \frac{B}{x},\qquad p(x,y) = \frac{x^{2} + x(1 + 2n + \alpha) + y}{N}.
\end{equation*}
Proceeding in the same way, from the correspondence $F_{3}-F_{4} = E_{1} - E_{2}$ we deduce that $N=s$, from the correspondence $F_{4} = E_{2}$ 
we get $B(1+n)=-sa_{1}$, and since we know, from the period map computation, that $a_{1} = n+1$, we get $B=-s$ which gives us
the required change of variables \eqref{eq:Laguerre-KNY-coord-change}. 
Finally, from the correspondence $F_{5} - F_{6} = E_{7} - E_{8}$ we get the relationship between the weight parameter $s$
and the Painlev\'e parameter $t$, $s = -t$. The inverse change of variables can now be either obtained directly, or computed in the similar way. This concludes
the proof of Theorem~\ref{thm:Laguerre-coord-change}.


\subsection{Half-maps and the point evolution} 
\label{sub:Laguerre-half-maps}
Consider now the individual maps $\tilde{\varphi}_{1}:(x,y)\to (x,\overline{y})$ and $\tilde{\varphi}_{2}:(x,y)\mapsto (\underline{x},y)$. Using the 
action of the half-maps on the Picard lattice described in Lemma~\ref{lem:Laguerre-Pic-action}, we see that the corresponding action on the symmetry roots 
$\alpha_{i}$, given by \eqref{eq:Laguerre-symmetry-rb}, is
\begin{align*}
	(\tilde{\varphi}_{1})_{*}&:\langle \alpha_{0},\alpha_{1},\alpha_{2},\alpha_{3} \rangle \mapsto \langle \alpha_{1},\alpha_{0},\alpha_{3},\alpha_{2} \rangle,\\
	(\tilde{\varphi}_{2})_{*}&:\langle \alpha_{0},\alpha_{1},\alpha_{2},\alpha_{3} \rangle \mapsto \langle -\alpha_{0},2\alpha_{0} + \alpha_{1},
	\alpha_{2} + 2\alpha_{3},-\alpha_{3} \rangle.
\end{align*}
The standard techniques (see, e.g., \cite{DzhTak:2015:GARFSTDPE}) then allow us to decompose these mappings in terms of generators of the extended affine Weyl 
symmetry group $\widetilde{W}(2A_{1}^{(1)})$ described in Section~\ref{sub:the_extended_affine_weyl_symmetry_group}. We get 
\begin{equation}
	\tilde{\varphi}_{1} \sim \sigma_{2}\sigma_{1}\sigma_{2}\sigma_{1}:(x,y)\to (\tilde{x},\tilde{y}),\qquad 
	\tilde{\varphi}_{2} \sim w_{0}w_{3}:(x,y)\to (\undertilde{x},\undertilde{y}).
\end{equation}
However, these mappings differ from the mappings $\tilde{\varphi}_{i}$ by some small gauge transformations. Specifically, $\tilde{x} = - x$, $\tilde{y}=\overline{y}$, 
$\undertilde{x} = - \underline{x}$, and $\undertilde{y} = y$.  For the full maps these gauge transformations cancel each other and the decomposition of, for example, 
the forward map is, as expected,
the same as the decomposition of the standard discrete Painlev\'e equation ${[1 \overline{1} \overline{1} 1]}$ given in \eqref{eq:dP-KNY-decomp},
\begin{align*}
	\tilde{\varphi} &= \tilde{\varphi}_{2}^{-1}\circ \tilde{\varphi}_{1} = w_{3}w_{0}\sigma_{2}\sigma_{1}\sigma_{2}\sigma_{1} = 
	w_{3}\sigma_{2}\sigma_{1}\sigma_{2}\sigma_{1} w_{1} \\
	&= \sigma_{2}\sigma_{1}\sigma_{2}\sigma_{1} w_{2}w_{1} 
	= \sigma_{1}\sigma_{2}\sigma_{1}\sigma_{2} w_{2}w_{1}:(x,y)\mapsto (\overline{x},\overline{y}), 
\end{align*}
where we used various commutation relations in the extended affine Weyl group $\widetilde{W}(2A_{1}^{(1)})$. However, these gauge transformations actually explain 
the evolution of points under half-maps given in Lemma~\ref{lem:Laguerre-Pic-action} parts (a) and (b). For example, under the forward half-map $\tilde{\varphi}_{1}$
the points $q_{2}$, $q_{4}$, and $q_{8}$ evolve as 
\begin{align*}
	&q_{2}(u_{1}=0,v_{1}=-n=a_{0})\mapsto \tilde{q}_{2}\left(\tilde{u}_{1}=-\overline{u}_{1}=0,
		\tilde{v}_{1} = - \overline{v}_{1} = - \overline{a}_{0} = -a_{1} = -(n+1)\right),\\
	&q_{4}(u_{3}=0,v_{3}=-(n+\alpha)=a_{3})\mapsto \tilde{q}_{4}\left(\tilde{u}_{3}=0, 
		\tilde{v}_{3} = - \overline{v}_{3} = - \overline{a}_{3} = -a_{2} = -(n+1 + \alpha)\right),\\
	&q_{8}(u_{7}=0,v_{7}= 1+ 2n+\alpha=a_{2}-a_{0})\mapsto 
	\tilde{q}_{8}\left(\tilde{u}_{7}=0,\right.\\ 
	&\hskip3in\left.\tilde{v}_{7}= - \overline{v}_{7} = - \overline{a}_{2} + \overline{a}_{0} = a_{1} - a_{3} = 1 + 2n + \alpha \right),
\end{align*}
and so it looks like points $q_{2}$ and $q_{4}$ evolve with the parameter $n$, and the point $q_{8}$ remains stationary. The situation with the backward 
half-map $\tilde{\varphi}_{2}$ is similar. This somewhat delicate point is explained at length in \cite[Section2.9]{DzhFilSto:2020:RCDOPWHWDPE}.



\section{Generalized Charlier Polynomials and the alt.~d-$\Pain{II}$ Equation} 
\label{sec:Charlier-weight}

In this section we consider another example of a recurrence relation that appeared in the study of orthogonal polynomials, this time
the polynomials with the generalized Charlier weight. 

Recall that, given some parameter $a>0$, the classical Charlier polynomials $C_{n}(x;a)$ (\cite[Chapter VI]{Chi:1978:IOP}, 
\cite[9.14]{Koekoek:2010aa}) 
are polynomials that are orthogonal on the lattice $\mathbb{N}=\{0,1,2,3,\ldots\}$ with
respect to the Poisson distribution: 
\begin{equation}
	\sum_{k=0}^{\infty}C_n(k;a)C_m(k;a)\frac{a^k}{k!}=a^{-n}\exp(a)\, n!\,\delta_{n,m},\qquad a>0.%
\end{equation}
The three term recurrence relation for the Charlier polynomials is
\begin{equation*}
-x C_n(x;a)=a C_{n+1}(x;a) - (n+a) C_n(x;a) + n C_{n-1}(x;a),
\end{equation*}
and the corresponding normalized recurrence relation for the \emph{monic} polynomials $P_{n}(x;a)$ is
\begin{equation*}
	x P_{n}(x;a) = P_{n+1}(x;a) + b_{n} P_{n}(x;a) + a_{n}^{2} P_{n-1}(x;a),\qquad C_{n}(x;a) = \left(-\frac{1}{a}\right)^{n} P_{n}(x;a),
\end{equation*}
where the recurrence coefficients are given explicitly  by $a_{n}^{2}=n a$ and $b_{n}=n+a$ 
for $n\in\mathbb{N}$. In \cite{SmeVan:2012:OPB} the classical Charlier weight was 
generalized by adding one more parameter $\beta>0$ to the weight function. The resulting function 
$w(x)=\frac{\Gamma(\beta)a^x}{\Gamma(\beta+x)\Gamma(x+1)}$, when restricted to the lattice $\mathbb{N}$, becomes
$w(k)=\frac{a^k}{(\beta)_k k!}$, where $(\cdot)_{k}$ is the usual Pochhammer symbol. Corresponding
monic orthogonal polynomials $P_{n}(x;a,\beta)$  satisfy 
\begin{equation}
	\label{eq:Charlier-ortho}
	\begin{aligned}
			&\sum_{k=0}^{\infty}P_n(k;a,\beta)P_m(k;a,\beta)\frac{a^k}{(\beta)_k k! }=0,\  n\neq m;\\
			&x P_n(x;a,\beta) = P_{n+1}(x;a,\beta) + b_{n} P_n(x;a,\beta) + a_{n}^{2} P_{n-1}(x;a,\beta).
	\end{aligned}
\end{equation} 
Recurrence coefficients $a_{n}^{2}$ and $b_{n}$ are described by the following Theorem.

\begin{theorem}{\rm \cite[Th.~2.1]{SmeVan:2012:OPB}}
	\label{thm:Charlier-rec-discrete}
	The recurrence coefficients for the orthogonal
polynomials defined by (\ref{eq:Charlier-ortho})  on the lattice $\mathbb{N}$  satisfy
\begin{equation}\label{eq:Charlier-rec-original}
	\begin{aligned}
		b_n+b_{n-1}-n+\beta&=\frac{an}{a_n^2},\\
		(a_{n+1}^2-a)(a_n^2-a)&=a(b_n-n)(b_n-n+\beta-1),
	\end{aligned}
\end{equation}
with the initial conditions
\begin{equation}\label{eq:Charlier-rec-initial}
a_0^2=0,\;\;b_0=\frac{\sqrt{a}
I_{\beta}(2\sqrt{a})}{I_{\beta-1}(2\sqrt{a})},
\end{equation} where
$I_{\beta}$ is the modified Bessel function of the first kind.
\end{theorem}

These generalized Charlier polynomials were further studied in \cite{FilVan:2013:RCGCPFPE}, 
where it was shown that the coefficients of the three-term recurrence relation
are related to solutions of the
(differential) fifth Painlev\'e equation $\Pain{V}$ 
\begin{equation}\label{eq:P5}
y''  =
\left(\frac{1}{2y}+\frac{1}{y-1}\right)(y')^2
-
\frac{y'}{t}+\frac{(y-1)^2}{t^2}\left(A
y + \frac{B}{y}\right)+\frac{C
y}{t}+\,\,\frac{D y (y+1)}{y-1},
\end{equation}
where $y=y(t)$ and $A,B,C,D$ are arbitrary complex parameters. In \cite{FilVan:2013:RCGCPFPE}, 
for the generalized Charlier case, these parameters were shown to be 
\begin{equation}\label{eq:P5-par}
A=\frac{(\beta-1)^2}{2},\;\;B=-\frac{(n+1)^2}{2},\;\;C=2k_1,\;\;D=0,
\end{equation}
where $t$ is related to the weight parameter $a$ as $a=e^{t}$ and $k_{1}$ is some scaling parameter. 
Moreover, it was shown that the recurrence \eqref{eq:Charlier-rec-original} was related to B\"acklund transformations of 
$\Pain{V}$. However, it is well-known, and it was also pointed out in \cite{FilVan:2013:RCGCPFPE}, that for 
$D=0$ the fifth Painlev\'e equation $\Pain{V}$ is in fact equivalent to the third Painlev\'e equation $\Pain{III}$,
\begin{equation}\label{eq:P3}
	u''= \frac{(u')^{2}}{u} - \frac{u'}{z} + \frac{\tilde{\alpha} u^{2} + \tilde{\beta}}{z} + \tilde{\gamma} u^{3} - \frac{\tilde{\delta}}{u},
\end{equation}
here $u=u(z)$ and $\tilde{\alpha},\tilde{\beta},\tilde{\gamma},\tilde{\delta}$ are again arbitrary complex constants. Thus, it is more natural to study the relationship 
between the recurrence \eqref{eq:Charlier-rec-original} and the B\"acklund transformations of $\Pain{III}$.

In this section we do just that. By studying \eqref{eq:Charlier-rec-original} from the geometric point of view, we show that it is
equivalent to a particular composition of elementary B\"acklund transformations, known as the alt.~d-$\Pain{II}$ equation. We
also show explicitly how to represent that recurrence as a composition of such elementary transformations, and also give an 
explicit change of coordinates reducing this recurrence to the standard form. 

As before, we first change the notation to $x_{n}:=a_{n}^{2}$, $y_{n}:=b_{n}$ and $\alpha:=a$.
Thus, our main object of study is the following discrete dynamical system:
\begin{equation}\label{eq:Charlier-xy-rec}
\left\{\begin{aligned}
(x_{n+1}-\alpha)(x_{n}-\alpha) &= \alpha(y_{n}-n)(y_{n}+\beta-n-1),\\
y_{n-1} + y_{n} +\beta - n &= \frac{\alpha n}{x_{n}}.
 \end{aligned}\right.
\end{equation}
The main result of this section is the following Theorem.

\begin{theorem}\label{thm:Charlier2Sakai-coord-change}
	Recurrence \eqref{eq:Charlier-xy-rec} is equivalent to Sakai's alt.~d-$\Pain{II}$ equation \eqref{eq:dP-Sakai-map}	
	\begin{equation*}
	{f}_{n+1}=b_{0} - a_{0} - f_{n} - g_{n} - \frac{s}{g_{n}},\qquad 
	{g}_{n+1} = \frac{s}{g_{n}} - \frac{b_{0}s}{g_{n}f_{n+1}},
	\end{equation*}
	via the following explicit change of coordinates and parameter identification:
	\begin{equation}\label{eq:Charlier2Sakai-coord}
		\left\{
		\begin{aligned}
			x_{n}(f,g) &=\frac{s(f_{n} + 1 - b_{0})}{g_{n}},\\
			y_{n}(f,g) &=f_{n} + g_{n},\\
			n&=b_{0}-1,\  \alpha = -s,\ \beta=a_{0},			
		\end{aligned}
		\right. 
		\qquad 
		\left\{
		\begin{aligned}
		f_{n}(x,y) &= \frac{x_{n}y_{n} - \alpha n}{x_{n} - \alpha},\\
		g_{n}(x,y) &= \frac{\alpha(n-y_{n})}{x_{n}-\alpha},\\
		a_0&=\beta,  \ a_1=1-\beta, \\
        b_0&=1+n,  \  b_{1}=-n,\  s=-\alpha.			
		\end{aligned}
		\right.
	\end{equation}
	Note that the evolution $n\to n+1$ corresponds to the correct root variables evolution $\overline{a}_{0} = a_{0}$, $\overline{a}_{1}=a_{1}$, 
	$\overline{b}_{0}=b_{0}+1$, $\overline{b}_{1} = b_{1} - 1$ (recall that $a_{0}+a_{1} = b_{0} + b_{1} = 1$).
\end{theorem}

The techniques used to establish this result are the same as in Section~\ref{sec:Laguerre-weight} and so here we only give a brief summary.

\subsection{Identification with alt.~d-$\Pain{II}$ equation} 
\label{sub:identification_with_alt_d_pain_ii_equation}

Recurrence \eqref{eq:Charlier-xy-rec} defines two half-step maps, $\widetilde{\psi}_{1}^{n}:(x_{n},y_{n})\to (x_{n+1},y_{n})$ and 
$\widetilde{\psi}_{2}^{n}:(x_{n},y_{n})\to (x_{n},y_{n-1})$. As usual, we put $x:= x_{n}$, $\overline{x}:=x_{n+1}$, $\underline{x}:=x_{n-1}$, and similarly
for $y$. We can then define the forward map 
$\widetilde{\psi}:=(\widetilde{\psi}_{2}^{n+1})^{-1} \circ \widetilde{\psi}_{1}^{n}:(x,y)\to (\overline{x},\overline{y})$
and the backward map 
$\widetilde{\psi}^{-1}:=(\widetilde{\psi}_{1}^{n-1})^{-1} \circ \widetilde{\psi}_{2}^{n}:(x,y)\to (\underline{x},\underline{y})$. These are the maps 
we study. 

Standard computation shows that the base points of the maps $\widetilde{\psi}$ and $\widetilde{\psi}^{-1}$ are 
 \begin{equation}\label{z-points}
	 \begin{aligned}
		    &z_{1}(x=\alpha,y=n),\quad z_{2}(x=\alpha,y=n+1-\beta),\\
		    &z_{3}(x=0,Y=0)\leftarrow z_{4}(U_{3}=xy=\alpha n,V_{3} = Y = 0),\\
		    &z_{5}(X=0,Y=0)\leftarrow z_{6}(U_{5}=Xy=0,V_{5}=Y=0)\leftarrow z_{7}(U_{6}=Xy^{2}=-1,V_{6}=0)\\
			&\qquad \leftarrow z_{8}(U_{7}=Xy^{3} + y=\beta-n,V_{7}=Y = 0),
	 \end{aligned}
 \end{equation}
and their configuration and the resulting surface are shown on Figure~\ref{fig:Charlier-pt-conf}. 

\begin{figure}[ht]
	\begin{center}		
	\begin{tikzpicture}[>=stealth,basept/.style={circle, draw=red!100, fill=red!100, thick, inner sep=0pt,minimum size=1.2mm}]
	\begin{scope}[xshift=0cm,yshift=0cm]
	\draw [black, line width = 1pt] (-0.2,0) -- (3.2,0)	node [pos=0,left] {\small $H_{y}$} node [pos=1,right] {\small $y=0$};
	\draw [black, line width = 1pt] (-0.2,3) -- (3.2,3) node [pos=0,left] {\small $H_{y}$} node [pos=1,right] {\small $y=\infty$};
	\draw [black, line width = 1pt] (0,-0.2) -- (0,3.2) node [pos=0,below] {\small $H_{x}$} node [pos=1,above] {\small $x=0$};
	\draw [black, line width = 1pt,dashed] (1,-0.2) -- (1,3.2) node [pos=0,below] {\small $H_{x}$} node [pos=1,above] {\small $x=\alpha$};
	\draw [black, line width = 1pt] (3,-0.2) -- (3,3.2) node [pos=0,below] {\small $H_{x}$} node [pos=1,above] {\small $x=\infty$};
	\node (z1) at (1,1) [basept,label={[xshift = 7pt, yshift=-15pt] \small $z_{1}$}] {};
	\node (z2) at (1,2) [basept,label={[xshift = 7pt, yshift=-15pt] \small $z_{2}$}] {};
	\node (z3) at (0,3) [basept,label={[xshift = 7pt, yshift=-15pt] \small $z_{3}$}] {};
	\node (z4) at (-0.5,2.5) [basept,label={[xshift = 7pt, yshift=-15pt] \small $z_{4}$}] {};
	\node (z5) at (3,3) [basept,label={[xshift = -7pt, yshift=-15pt] \small $z_{5}$}] {};
	\node (z6) at (3,2.3) [basept,label={[right] \small $z_{6}$}] {};
	\node (z7) at (3.5,1.8) [basept,label={[right] \small $z_{7}$}] {};
	\node (z8) at (3.5,1.3) [basept,label={[right] \small $z_{8}$}] {};
	\draw [red, line width = 0.8pt, ->] (z4) -- (z3);
	\draw [red, line width = 0.8pt, ->] (z6) -- (z5);
	\draw [red, line width = 0.8pt, ->] (z7) -- (z6);
	\draw [red, line width = 0.8pt, ->] (z8) -- (z7);
	\end{scope}
	\draw [->] (6.5,1.5)--(4.55,1.5) node[pos=0.5, below] {$\operatorname{Bl}_{z_{1}\cdots z_{8}}$};
	\begin{scope}[xshift=8.5cm,yshift=0cm]
	\draw [line width = 1pt] (-0.2,0) -- (3.7,0)	node [pos=0,xshift=-8pt, yshift=0pt] {\small $H_{y}$} {};
	\draw [blue, line width = 1pt] (0.3,3) -- (3.7,3)	node [pos=0,xshift=-25pt, yshift=2pt] {\small $H_{y}-L_{35}$} {};
	\draw [red, line width = 1pt] (0,-0.2) -- (0,2.7) node [pos=0,xshift=-5pt, yshift=-7pt] {\small $H_{x}-L_{3}$};
	\draw[red, line width = 1pt] (1,0.2)--(2,1.2) node [pos=1,xshift=3pt, yshift=5pt] {\small $L_{1}$};
	\draw[red, line width = 1pt] (1,1.2)--(2,2.2) node [pos=1,xshift=3pt, yshift=5pt] {\small $L_{2}$};
	\draw[red, line width = 1pt] (-0.2,0.4)--(1.4,0.4) node [pos=0,xshift=-20pt, yshift=0pt] {\small $H_{y} - L_{1}$};
	\draw[red, line width = 1pt] (-0.2,1.4)--(1.4,1.4) node [pos=0,xshift=-20pt, yshift=0pt] {\small $H_{y} - L_{2}$};
	\draw[blue, line width = 1pt] (1.5,-0.2)--(1.5,3.2) node [pos=0,xshift=3pt, yshift=-7pt] {\small $H_{x} - L_{12}$};
	\draw[blue, line width = 1pt] (3.2,-0.2)--(4.2,0.8) node [pos=0, xshift=15pt, yshift=-7pt] {\small $H_{x} - L_{56}$};
	\draw[blue, line width = 1pt] (-0.2,2.2)--(0.8,3.2) node [pos=0,xshift=-20pt, yshift=-2pt] {\small $L_{3}-L_{4}$};
	\draw[red, line width = 1pt] (0.1,2.9)--(1.1,1.9) node [pos=1,xshift=-5pt, yshift=-3pt] {\small $L_{4}$};	
	\draw[blue, line width = 1pt] (3.2,3.2)--(4.2,2.2) node [pos=0,xshift=0pt, yshift=5pt] {\small $L_{5}- L_{6}$};
	\draw[blue, line width = 1pt] (4,0.3)--(4,2.7) node [pos=1,xshift=15pt, yshift=7pt] {\small $L_{6}- L_{7}$};
	\draw[blue,  line width = 1pt] (3.5,1.5)--(5,1.5) node[pos = 0, xshift=-10pt, yshift=7pt]  {\small $L_{7}- L_{8}$};
	\draw[red, line width = 1pt ] (4.7,2)--(4.7,0.5) node[pos=1,xshift=3pt, yshift=-7pt] {\small $L_{8}$};
	\end{scope}
	\end{tikzpicture}
	\end{center}
	\caption{The Sakai Surface for the generalized Charlier recurrence} 
	\label{fig:Charlier-pt-conf}
\end{figure}	

We then immediately see that we get a $D_{6}^{(1)}$ Sakai surface,
with the surface root basis shown on Figure~\ref{fig:Charlier-surface-roots}.
\begin{figure}[ht]
	\begin{equation}\label{eq:Charlier-surface-rb}
		\raisebox{-32.1pt}{\begin{tikzpicture}[
			elt/.style={circle,draw=black!100,thick, inner sep=0pt,minimum size=2mm}]
			\path
			(-1,1)
			node (d0) [elt, label={[xshift=-10pt, yshift = -10 pt] $\delta_{0}$} ] {}
			(-1,-1)
			node (d1) [elt, label={[xshift=-10pt, yshift = -10 pt] $\delta_{1}$} ] {}
			( 0,0)
			node   (d2) [elt, label={[xshift=0pt, yshift = -25 pt] $\delta_{2}$} ] {}
			( 1,0)
			node   (d3) [elt, label={[xshift=0pt, yshift = -25 pt] $\delta_{3}$} ] {}
			( 2,0)
			node   (d4) [elt, label={[xshift=0pt, yshift = -25 pt] $\delta_{4}$} ] {}
			( 3,1)
			node   (d6) [elt, label={[xshift=10pt, yshift = -10 pt] $\delta_{6}$} ] {}
			( 3,-1) node (d5) [elt, label={[xshift=10pt, yshift = -10 pt] $\delta_{5}$} ] {};
			\draw [black,line width=1pt ] (d0) -- (d2) -- (d1)  (d2) -- (d3) --(d4) (d6) -- (d4) -- (d5);
			\end{tikzpicture}} \qquad
		\begin{alignedat}{2}
			\delta_{0} &= \mathcal{H}_{x} - \mathcal{L}_{56}, &\qquad  \delta_{4} &=\mathcal{H}_{y} - \mathcal{L}_{35},\\
			\delta_{1} &= \mathcal{L}_{7}-\mathcal{L}_{8}, &\qquad  \delta_{5} &= \mathcal{H}_{x} - \mathcal{L}_{12},\\
			\delta_{2} &= \mathcal{L}_{6}-\mathcal{L}_{7}, &\qquad  \delta_{6} &= \mathcal{L}_{3} - \mathcal{L}_{4}.\\
			\delta_{3} &= \mathcal{L}_{5}-\mathcal{L}_{6}, &\qquad
		\end{alignedat}
	\end{equation}
	\caption{The Surface Root Basis for the generalized Charlier recurrence}
	\label{fig:Charlier-surface-roots}
\end{figure}

The surface root basis \eqref{eq:Charlier-surface-rb} corresponds to the change of basis between the Sakai surface on Figure~\ref{fig:Charlier-pt-conf}
and the alt.~d-$\Pain{II}$ surface on Figure~\ref{fig:Sakai-pt-conf} given by 
	\begin{equation}\label{eq:Charlier-Sakai-Pic}
		\begin{aligned}
				\mathcal{H}_{x} &=\mathcal{H}_{f} + \mathcal{H}_{g} - \mathcal{K}_{16}, &\qquad 
				\mathcal{H}_{f} &= \mathcal{H}_{x} + \mathcal{H}_{y} - \mathcal{L}_{13}, \\
				\mathcal{H}_{y} &=\mathcal{H}_{f} + \mathcal{H}_{g} - \mathcal{K}_{67}, &\qquad 
				\mathcal{H}_{g} &= \mathcal{H}_{x} + \mathcal{H}_{y} - \mathcal{L}_{15}, \\
				\mathcal{L}_{1} &=\mathcal{H}_{f} + \mathcal{H}_{g} - \mathcal{K}_{167}, &\qquad  
				\mathcal{K}_{1} &= \mathcal{H}_{y}  - \mathcal{L}_{1},\\
				\mathcal{L}_{2} &=\mathcal{K}_{8}, &\qquad  
				\mathcal{K}_{2} &= \mathcal{L}_{6},\\		
				\mathcal{L}_{3} &=\mathcal{H}_{g} - \mathcal{K}_{6}, &\qquad  
				\mathcal{K}_{3} &= \mathcal{L}_{7},\\		
				\mathcal{L}_{4} &=\mathcal{K}_{5}, &\qquad  
				\mathcal{K}_{4} &= \mathcal{L}_{8},\\		
				\mathcal{L}_{5} &=\mathcal{H}_{f} - \mathcal{K}_{6}, &\qquad  
				\mathcal{K}_{5} &= \mathcal{L}_{4},\\		
				\mathcal{L}_{6} &=\mathcal{K}_{2}, &\qquad  
				\mathcal{K}_{6} &= \mathcal{H}_{x} + \mathcal{H}_{y} - \mathcal{L}_{135},\\		
				\mathcal{L}_{7} &=\mathcal{K}_{3}, &\qquad  
				\mathcal{K}_{7} &= \mathcal{H}_{x} - \mathcal{L}_{1},\\		
				\mathcal{L}_{8} &=\mathcal{K}_{4}, &\qquad  
				\mathcal{K}_{8} &= \mathcal{L}_{2}.
		\end{aligned}
	\end{equation}

Then the symmetry roots are given by 
\begin{equation}\label{eq:Charlier-symmetry-rb}
	\begin{aligned}
		\alpha_{0} &= 2\mathcal{H}_{x} + 2\mathcal{H}_{y} - 2\mathcal{L}_{1} - \mathcal{L}_{345678},&\qquad
		\alpha_{2} &= \mathcal{H}_{x} + 2\mathcal{H}_{y} -  \mathcal{L}_{125678},\\ 
		\alpha_{1} &= \mathcal{L}_{1} - \mathcal{L}_{2},&\qquad
		\alpha_{3} &= \mathcal{H}_{x} - \mathcal{L}_{34}.
	\end{aligned}
\end{equation}
	
The map $\tilde{\psi}$ induces the evolution $\tilde{\psi}_{*}: \operatorname{Pic}(\mathcal{X}) \to \operatorname{Pic}(\overline{\mathcal{X}})$ given by 
	\begin{equation}
		\begin{aligned}
			\mathcal{H}_{x}&\mapsto 3 \overline{\mathcal{H}}_{x} + 2 \overline{\mathcal{H}}_{y} - \overline{\mathcal{L}}_{12}
			-  2\overline{\mathcal{L}}_{34}  - \overline{\mathcal{L}}_{56}, &\qquad 
			\mathcal{L}_{4}&\mapsto \overline{\mathcal{L}}_{8},\\
			\mathcal{H}_{y}&\mapsto \overline{\mathcal{H}}_{x} + \overline{\mathcal{H}}_{y} - \overline{\mathcal{L}}_{34}, &\qquad 
			\mathcal{L}_{5}&\mapsto \overline{\mathcal{H}}_{x} + \overline{\mathcal{H}}_{y} - \overline{\mathcal{L}}_{346},\\
			\mathcal{L}_{1}&\mapsto \overline{\mathcal{H}}_{x} + \overline{\mathcal{H}}_{y} - \overline{\mathcal{L}}_{234}, &\qquad 
			\mathcal{L}_{6}&\mapsto \overline{\mathcal{H}}_{x} + \overline{\mathcal{H}}_{y} - \overline{\mathcal{L}}_{345},\\
			\mathcal{L}_{2}&\mapsto \overline{\mathcal{H}}_{x} + \overline{\mathcal{H}}_{y} - \overline{\mathcal{L}}_{134}, &\qquad 
			\mathcal{L}_{7}&\mapsto \overline{\mathcal{H}}_{x} - \overline{\mathcal{L}}_{4},\\
			\mathcal{L}_{3}&\mapsto \overline{\mathcal{L}}_{7}, &\qquad 
			\mathcal{L}_{8}&\mapsto \overline{\mathcal{H}}_{x} - \overline{\mathcal{L}}_{3},\\
		\end{aligned}
	\end{equation}
	which then gives the translation 
	\begin{equation}\label{eq:Charlier-trans}
	\tilde{\psi}_*:\upalpha=\langle\alpha_0,\alpha_1,\alpha_2,\alpha_3\rangle\mapsto
	\tilde{\psi}_*(\upalpha)=\upalpha+ \langle 0,0,-1,1 \rangle\delta
	\end{equation}
	on the symmetry root lattice. Thus, we see that the dynamics is equivalent to that of the alt.~d-$\Pain{II}$ equation, and to prove 
	Theorem~\ref{thm:Charlier2Sakai-coord-change} it remains to find the change of variables inducing the change of basis \eqref{eq:Charlier-Sakai-Pic},
	which is done in the usual way.



\section{Conclusions} 
\label{sec:conclusions}
In this paper we illustrated the effectiveness of the algorithmic approach, recently proposed in \cite{DzhFilSto:2020:RCDOPWHWDPE}, for solving the 
discrete Painlev\'e identification problem using the algebro-geometric tools of Sakai's geometric theory of Painlev\'e equations. 
We emphasized the importance of understanding not just the type of the surface, but also the actual translation element, by considering 
two discrete Painlev\'e equations, that appeared in applied problems, and that are \emph{non-equivalent} but at the same time are 
regularized on the \emph{same $D_{6}^{(1)}$ surface family}. We also explained how the geometric approach 
allows us to quickly and completely understand each dynamic, decompose it as a composition of elementary B\"acklund transformations, and also 
match the parameters of the problem with the canonical Painlev\'e parameters, i.e., the root variables.

A very interesting problem is to construct a natural degeneration scheme for weights for orthogonal polynomial ensembles 
that corresponds to the degeneration cascade for Sakai's classification scheme for discrete Painlev\'e equations. The examples we considered in this paper 
show that any such scheme has to be more refined than just a geometric classification, and it should indeed include the class of translation elements. 
It is interesting and important to collect more examples of this type to better understand how such problem can be approached. 

\section*{Acknowledgements} 
\label{sec:acknowledgements} 
XL is supported by the National Natural Science Foundation of China (No. 12301309). DJZ is supported by the National Natural Science Foundation of China (No. 12271334, 12326428) and Science and Technology Innovation  Plan of Shanghai (No. 20590742900) from Science and Technology Commission of Shanghai Municipality. GF acknowledges the support of the 
National Science Center (Poland) grant OPUS 2017/25/B/BST1/00931. 

\section*{Statements and Declarations} 
On behalf of all authors, the corresponding author states that there is no
conflict of interest and data sharing is not applicable to this article as no datasets were generated or analysed during the current study.

\appendix 

\section{Discrete Painlev\'e Equations in the d-$\dPain{2A_{1}^{(1)}/D_{6}^{(1)}}$ Family} 
\label{sec:dP-2A1-std}


Let us very briefly recall some main ingredients of Sakai's classification scheme \cite{Sak:2001:RSAWARSGPE} for discrete Painlev\'e equations. Each such 
equation describes a discrete dynamical system on a certain family of rational algebraic surfaces obtained by blowing up eight points, called the \emph{base points}, 
on $\mathbb{P}^{1}_{\mathbb{C}} \times  \mathbb{P}^{1}_{\mathbb{C}}$. For each family, the configuration of these base points is constrained to stay on 
a collection of irreducible rational (except for the elliptic case) curves -- these curves are the irreducible components of the 
(unique) anti-canonical divisor $-K_{\mathcal{X}}$ that is fixed for each family. Modulo gauging, the locations of the base points within that configuration 
are the parameters for the family; there exist canonical such parameters known as the \emph{root variables}. 
It turns out that the intersection configuration of these irreducible components is described by some \emph{affine Dynkin diagram} 
$\mathcal{D}_{1}$ (in our case, $D_{6}^{(1)}$) and the type of this diagram is known as the \emph{geometric type} of the family. Symmetries of this family of surfaces 
form an affine Weyl group defined by some other Dynkin diagram $\mathcal{D}_{2}$ (in our case, $2A_{1}^{(1)}$) further extended by some diagram automorphisms. The type 
of the diagram $\mathcal{D}_{2}$ is known as the \emph{symmetry type} of the family. The dynamics itself is generated by some composition of elementary symmetries and 
corresponds to a translation element in extended affine Weyl group $\widetilde{W}(\mathcal{D}_{2})$. The conjugacy class of that element is then the type of our 
discrete Painlev\'e equation, it is a complete invariant. 

Thus, we see that for each surface there are infinitely many non-equivalent discrete Painlev\'e dynamics. Moreover, each equation has many different coordinate realizations,
since there are different ways to define point configurations of the given geometric type $\mathcal{D}_{1}$. The standard and most studied examples of discrete 
Painlev\'e equations correspond to some particularly simple point configurations and ``short'' translations. Recent comprehensive survey paper \cite{KajNouYam:2017:GAPE}
gave standard examples of such configurations for each type, together with natural degenerations, and we use it as our main reference. 

For the d-$\dPain{2A_{1}^{(1)}/D_{6}^{(1)}}$ family, the standard example in \cite{KajNouYam:2017:GAPE} is equation (8.29) in Section 8.1.20. Based on the action 
of this dynamics on the root variables we label it as ${[1\overline{1}\overline{1}1]}$. However, there is also another well-known and non-equivalent example of 
discrete Painlev\'e equation of this type, known as alt.~d-$\Pain{II}$, in \cite{Sak:2001:RSAWARSGPE}. We use a sightly different form of this equation as given by
equations (2.35) (2.36) in \cite{Sak:2007:PDPETLF} and label it as ${[001\overline{1}]}$. In this appendix we describe the 
point configurations for both of these equations (and fix some minor typos in \cite{KajNouYam:2017:GAPE}).

\begin{remark} To see that equations ${[1\overline{1}\overline{1}1]}$ and ${[001\overline{1}]}$ are indeed non-equivalent one can, for example, 
	compute the length of the corresponding translation vectors on the lattice. But probably the simplest way is to 
	look at the Jordan block structure of the induced linear map on the Picard lattice. For equation ${[1\overline{1}\overline{1}1]}$
	it is $J(-1,1)^{\oplus2}\oplus J(1,1)^{\oplus5}\oplus J(1,3)$ and for equation ${[001\overline{1}]}$ it is 
	$J(-1,1)^{\oplus3}\oplus J(1,1)^{\oplus4}\oplus J(1,3)$.
\end{remark}

\subsection{The point configuration}\label{app:point-conf}
Recall that the Picard lattice of a rational algebraic surface $\mathcal{X}$ obtained by blowing up 
$\mathbb{P}^{1}_{\mathbb{C}} \times  \mathbb{P}^{1}_{\mathbb{C}}$ at eight points is generated by the 
classes of coordinate lines $\mathcal{H}_{1}$ and $\mathcal{H}_{2}$ and the classes $\mathcal{E}_{i}$
of the central fibers of the blowup,
\begin{equation*}
	\operatorname{Pic}_{\mathcal{X}} = \operatorname{Span}_{\mathbb{Z}}\{\mathcal{H}_{1},\mathcal{H}_{2},\mathcal{E}_{1},\ldots,\mathcal{E}_{8}\}.
\end{equation*}
This lattice is equipped with the symmetric bilinear product  (\emph{intersection form}) defined on the generators by 
$\mathcal{H}_{1}\bullet \mathcal{H}_{2} = 1$, 
$\mathcal{H}_{1}\bullet \mathcal{H}_{1} = \mathcal{H}_{2}\bullet \mathcal{H}_{2} = \mathcal{H}_{i}\bullet \mathcal{E}_{j} = 0$, and 
$\mathcal{E}_{i}\bullet \mathcal{E}_{j} = -\delta_{ij}$. In this lattice we have the anti-canonical divisor class 
$\mathcal{K}_{\mathcal{X}} = 2 \mathcal{H}_{1} + 2 \mathcal{H}_{2} - \mathcal{E}_{1} - \cdots -\mathcal{E}_{8}$
which for the $D_{6}^{(1)}$-type surface should decompose in the irreducible $-2$ components as
$\mathcal{K}_{\mathcal{X}} = \delta = \delta_{0} + \delta_{1} + 2 \delta_{2} + 2 \delta_{3} + 2\delta_{4} + \delta_{5} + \delta_{6}$. 
In \cite{KajNouYam:2017:GAPE} this decomposition is achieved by the choice of the \emph{surface root basis} $R = \{\delta_{i}\}$ 
shown on Figure~\ref{fig:KNY-surface-roots}.
\begin{figure}[ht]
	\begin{equation}\label{eq:std-surface-rb}
		\raisebox{-32.1pt}{\begin{tikzpicture}[
			elt/.style={circle,draw=black!100,thick, inner sep=0pt,minimum size=2mm}]
			\path
			(-1,1)
			node (d0) [elt, label={[xshift=-10pt, yshift = -10 pt] $\delta_{0}$} ] {}
			(-1,-1)
			node (d1) [elt, label={[xshift=-10pt, yshift = -10 pt] $\delta_{1}$} ] {}
			( 0,0)
			node   (d2) [elt, label={[xshift=0pt, yshift = -25 pt] $\delta_{2}$} ] {}
			( 1,0)
			node   (d3) [elt, label={[xshift=0pt, yshift = -25 pt] $\delta_{3}$} ] {}
			( 2,0)
			node   (d4) [elt, label={[xshift=0pt, yshift = -25 pt] $\delta_{4}$} ] {}
			( 3,1)
			node   (d6) [elt, label={[xshift=10pt, yshift = -10 pt] $\delta_{6}$} ] {}
			( 3,-1) node (d5) [elt, label={[xshift=10pt, yshift = -10 pt] $\delta_{5}$} ] {};
			\draw [black,line width=1pt ] (d0) -- (d2) -- (d1)  (d2) -- (d3) --(d4) (d6) -- (d4) -- (d5);
			\end{tikzpicture}} \qquad
		\begin{alignedat}{2}
			\delta_{0} &= \mathcal{E}_1-\mathcal{E}_2, &\qquad  \delta_{4} &=\mathcal{E}_6-\mathcal{E}_7,\\
			\delta_{1} &= \mathcal{E}_3-\mathcal{E}_4, &\qquad  \delta_{5} &= \mathcal{E}_5-\mathcal{E}_6,\\
			\delta_{2} &= \mathcal{H}_1-\mathcal{E}_{13}, &\qquad  \delta_{6} &= \mathcal{E}_7-\mathcal{E}_8.\\
			\delta_{3} &= \mathcal{H}_2-\mathcal{E}_{56}, &\qquad
		\end{alignedat}
	\end{equation}
	\caption{The Surface Root Basis for the standard  d-$\dPain{2A_1^{(1)}/D_6^{(1)}}$ point configuration}
	\label{fig:KNY-surface-roots}
\end{figure}

Let us now describe the corresponding point configuration. Take $(q,p)$ to be the coordinates in the $\mathbb{C}\times \mathbb{C}$ affine chart
and let $Q = 1/q$ and $P = 1/p$ be the coordinates at infinity. Looking at the surface root basis, we see that there are two points $p_{1}$
and $p_{3}$ on the vertical line that we can, using the M\"obius group action on coordinates, take to be $q=\infty$ (or $Q=0$). 
There are also two degeneration cascades (infinitely close points) $p_{2}\to p_{1}$ and $p_{4}\to p_{3}$. There is also a longer degeneration 
cascade $p_{8}\to p_{7}\to p_{6} \to p_{5}$. Using M\"obius group action again, we can arrange the points to be $p_{1}(\infty,1)$, 
$p_{3}(\infty,0)$, $p_{5}(0,\infty)$, and the only remaining gauge action is the rescaling in the $q$-coordinate (that we later use to normalize the root variables). 
We then get the point configuration 
shown on Figure~\ref{fig:KNY-pt-conf}. Note that one reason for choosing
this point normalization is that they are located on the polar divisor of the standard symplectic form $\omega = dp\wedge dq$. 

\begin{figure}[ht]
	\begin{center}		
	\begin{tikzpicture}[>=stealth,basept/.style={circle, draw=red!100, fill=red!100, thick, inner sep=0pt,minimum size=1.2mm}]
	\begin{scope}[xshift=0cm,yshift=0cm]
	\draw [black, line width = 1pt] (-0.2,0) -- (3.2,0)	node [pos=0,left] {\small $H_{p}$} node [pos=1,right] {\small $p=0$};
	\draw [black, line width = 1pt] (-0.2,3) -- (3.2,3) node [pos=0,left] {\small $H_{p}$} node [pos=1,right] {\small $p=\infty$};
	\draw [black, line width = 1pt] (0,-0.2) -- (0,3.2) node [pos=0,below] {\small $H_{q}$} node [pos=1,xshift = -7pt, yshift=5pt] {\small $q=0$};
	\draw [black, line width = 1pt] (3,-0.2) -- (3,3.2) node [pos=0,below] {\small $H_{q}$} node [pos=1,xshift = 7pt, yshift=5pt] {\small $q=\infty$};
	\node (p1) at (3,0) [basept,label={[xshift = -7pt, yshift=-3pt] \small $p_{3}$}] {};
	\node (p2) at (3.5,0.5) [basept,label={[xshift = 10pt, yshift=-7pt] \small $p_{4}$}] {};
	\node (p3) at (3,1) [basept,label={[xshift = -7pt, yshift=-3pt] \small $p_{1}$}] {};
	\node (p4) at (3.5,1.5) [basept,label={[xshift = 10pt, yshift=-7pt] \small $p_{2}$}] {};
	\node (p5) at (0,3) [basept,label={[xshift = 7pt, yshift=-15pt] \small $p_{5}$}] {};
	\node (p6) at (0.7,3) [basept,label={[yshift=0pt] \small $p_{6}$}] {};
	\node (p7) at (1.2,3.5) [basept,label={[yshift=-3pt] \small $p_{7}$}] {};
	\node (p8) at (1.9,3.5) [basept,label={[yshift=-3pt] \small $p_{8}$}] {};
	\draw [red, line width = 0.8pt, ->] (p2) -- (p1);
	\draw [red, line width = 0.8pt, ->] (p4) -- (p3);
	\draw [red, line width = 0.8pt, ->] (p6) -- (p5);
	\draw [red, line width = 0.8pt, ->] (p7) -- (p6);
	\draw [red, line width = 0.8pt, ->] (p8) -- (p7);
	\end{scope}
	\draw [->] (7,1.5)--(5,1.5) node[pos=0.5, below] {$\operatorname{Bl}_{p_{1}\cdots p_{8}}$};
	\begin{scope}[xshift=8.5cm,yshift=0cm]
	\draw[blue, line width = 1pt] (-0.5,3)--(3.5,3) node [pos=1,right] {\small $H_p- E_{56}$};
	\draw[red, line width = 1pt] (0,-0.2) -- (-1,0.8) node [pos=0,below] {\small $H_{q}-E_{5}$};	
	\draw[blue, line width = 1pt] (3,0.3)--(3,3.3) node [pos=1,above] {\small $ H_q- E_{13}$};
	\draw[red, line width = 1pt] (-0.7,0) -- (2.7,0) node [pos=0,left] {\small $H_{p}-E_{3}$};	
	\draw[blue, line width = 1pt] (2.2,-0.2)--(3.2,0.8); \node[blue] at (2.2,-0.5) {\small $ E_3- E_4$};
	\draw[red, line width = 1pt] (3.2,-0.2) -- (2.2,0.8) node [pos=1,left] {\small $E_{4}$};	
	\draw[blue,  line width = 1pt] (2.5,1.5)--(4,1.5);\node[blue] at (4.8,1.5) {\small $E_1- E_2$};
	\draw[red, line width = 1pt ] (3.7,2)--(3.7,0.5); \node[red] at (3.7,0.2) {\small $E_2$};
	\draw[blue, line width = 1pt] (-0.8,0.3)--(-0.8,2.7) node [pos=1,above left] {\small $E_5- E_6$};
	\draw[blue, line width = 1pt] (0,3.2)--(-1,2.2) node [pos=0,above]  {\small $E_6- E_7$};
	\draw[blue, line width = 1pt] (-0.6,2.9)--(0.4,1.9) node [pos=1,below] {\small $E_7- E_8$};
	\draw[red, , line width = 1pt] (-0.5,1.8)--(0.5,2.8) node [pos=1,right] {\small $E_8$};
	\end{scope}
	\end{tikzpicture}
	\end{center}
	\caption{The model Sakai Surface for the d-P$(2A_1^{(1)}/D_6^{(1)})$  example} 
	\label{fig:KNY-pt-conf}
\end{figure}

The parameterization of this point configuration in terms of root variables $a_{0},\ldots,a_{3}$ normalized by $a_{0}+a_{1} = a_{2}+a_{3} = 1$
is given in \cite{KajNouYam:2017:GAPE},
     \begin{equation*}
     p_{12}:\left(\frac{1}{\varepsilon},1-a_1\varepsilon \right)_2,\quad p_{34}:\left(\frac{1}{\varepsilon},-a_2\varepsilon \right)_2,\quad 	p_{5678}:\left(\varepsilon,-\frac{t}{\varepsilon^2}+\frac{1-a_1-a_2}{\varepsilon}\right)_4,
     \end{equation*}
	 or, explicitly,
     \begin{equation}\label{p-points}
     \begin{split}
     p_1(Q=0,p=1)&\leftarrow p_2(u_1=Q=0,v_1=q(p-1)=-a_1),\\
     p_3(Q=0,p=0)&\leftarrow p_4(u_3=Q=0,v_3=qp=-a_2),\\
     p_5(q=0,P=0)&\leftarrow p_6(u_5=q=0,v_5=PQ=0)\leftarrow p_7\left(u_6=q=0,v_6=\frac{1}{q^{2} p}=-\frac{1}{t}\right)\\
     &\leftarrow p_8\left(u_7=q=0,v_7=\frac{q^{2} p + t}{q^{3} p t}=\frac{a_1+a_2-1}{t^2}\right).
     \end{split}
     \end{equation}
The parameter $t$ that appears in the $p_{5}$ cascade is related to the independent variable of the differential Painlev\'e equation $\Pain{III}$ for which
this surface family, after we remove the polar divisor, is the Okamoto space of initial conditions. 
	 
\subsection{The extended affine Weyl symmetry group} 
\label{sub:the_extended_affine_weyl_symmetry_group}

Consider now the symmetry root basis $R^{\perp} = \{\alpha_{j}\}$ where the symmetry roots $\alpha_{j}$ 
should satisfy the condition $\delta_{i}\bullet \alpha_{j} = 0$. We get a reducible affine Dynkin diagram 
of type $2A_{1}^{(1)}$ (here $\mathcal{E}_{i\cdots j} = \mathcal{E}_{i}+\cdots+\mathcal{E}_{j}$), as shown in Figure \ref{fig:standard-symmetry-roots}.
\begin{figure}[ht]
	\begin{equation}\label{eq:std-symm-rb}
	\begin{tikzpicture}[baseline=-0.6cm]		
		\node (a0) at (0,0) [circle, thick, draw=black!100, inner sep=0pt,minimum size=1.3ex] {};
		\node (a1) at (1,0) [circle, thick, draw=black!100, inner sep=0pt,minimum size=1.3ex] {};
		\draw[thick, double distance = .4ex] (a0) -- (a1);
		\node[below] at (a0.south) {$\alpha_{0}$};
		\node[below] at (a1.south) {$\alpha_{1}$};
		\node (a2) at (0,-0.8) [circle, thick, draw=black!100, inner sep=0pt,minimum size=1.3ex] {};
		\node (a3) at (1,-0.8) [circle, thick, draw=black!100, inner sep=0pt,minimum size=1.3ex] {};
		\draw[thick, double distance = .4ex] (a2) -- (a3);
		\node[below] at (a2.south) {$\alpha_{2}$};
		\node[below] at (a3.south) {$\alpha_{3}$};
	\end{tikzpicture}\qquad \qquad 
	\begin{aligned}
		\alpha_{0} &= 2\mathcal{H}_q+\mathcal{H}_p-\mathcal{E}_{345678},
		&\quad \alpha_{1}&= \mathcal{H}_p-\mathcal{E}_{1} - \mathcal{E}_{2},\\
		\alpha_{2} &= \mathcal{H}_p-\mathcal{E}_{3} - \mathcal{E}_{4},& \quad 
		\alpha_{3}&= 2\mathcal{H}_q+\mathcal{H}_p-\mathcal{E}_{125678}.	\\
		-\mathcal{K}_{\mathcal{X}} &= \delta = \alpha_{0} + \alpha_{1} = \alpha_{2} + \alpha_{3}.	
	\end{aligned}								
	\end{equation}
	\caption{The Symmetry Root Basis for the standard  d-$\dPain{2A_1^{(1)}/D_6^{(1)}}$ point configuration}
	\label{fig:standard-symmetry-roots}
\end{figure}

The symmetry sub-lattice is $Q = \operatorname{Span}_{\mathbb{Z}}\{R^{\perp}\} = \operatorname{Span}_{\mathbb{Z}}\{\alpha_{j}\}$, and the
root variables $a_{i}$ are defined using the \emph{Period Map} $\chi:Q\to \mathbb{C}$, $a_{i} = \chi(\alpha_{i})$, together with the normalization
$\chi(\delta) = a_{0} + a_{1} = a_{2} + a_{3} = 1$, see example of detailed computation of the Period Map in \cite{DzhTak:2018:SASGTDPE,DzhFilSto:2020:RCDOPWHWDPE}.

The affine Weyl group $W\left(2A_{1}^{(1)}\right)$ is defined in terms of generators $w_{i} = w_{\alpha_{i}}$ and relations that 
are encoded by the affine Dynkin diagram $2A_{1}^{(1)}$,
\begin{equation*}
	W\left(2A_{1}^{(1)}\right) = W\left(\raisebox{-18pt}{\begin{tikzpicture}[
			elt/.style={circle,draw=black!100,thick, inner sep=0pt,minimum size=1.3ex},scale=0.8]
		\path 	(0,0) 	node 	(a0) [elt, label={[xshift=0pt, yshift = -20 pt] $\alpha_{0}$} ] {}
		        (1,0) node 	(a1) [elt, label={[xshift=0pt, yshift = -20 pt] $\alpha_{1}$} ] {}
		        (0,-1) 	node  	(a2) [elt, label={[xshift=0pt, yshift = -20 pt] $\alpha_{2}$} ] {}
		        (1,-1) 	node  	(a3) [elt, label={[xshift=0pt, yshift = -20 pt] $\alpha_{3}$} ] {};
		\draw [black,thick, double distance = .4ex] (a0) -- (a1)  (a2) -- (a3); 
	\end{tikzpicture}} \right)
	=
	\left\langle w_{0},\dots, w_{3}\ \left|\ 
	\begin{alignedat}{2}
    w_{i}^{2} = e,\quad  w_{i}\circ w_{j} &= w_{j}\circ w_{i}& &\text{ when 
   				\raisebox{-0.1in}{\begin{tikzpicture}[
   							elt/.style={circle,draw=black!100,thick, inner sep=0pt,minimum size=1.5mm}]
   						\path   ( 0,0) 	node  	(ai) [elt] {}
   						        ( 0.5,0) 	node  	(aj) [elt] {};
   						\draw [black] (ai)  (aj);
   							\node at ($(ai.south) + (0,-0.2)$) 	{$\alpha_{i}$};
   							\node at ($(aj.south) + (0,-0.2)$)  {$\alpha_{j}$};
   							\end{tikzpicture}}}
	\end{alignedat}\right.\right\rangle. 
\end{equation*} 
This group is realized via actions on $\operatorname{Pic}(\mathcal{X})$ given by reflections in the 
roots $\alpha_{i}$, 
\begin{equation}\label{eq:root-refl}
	w_{i}(\mathcal{C}) = w_{\alpha_{i}}(\mathcal{C}) = \mathcal{C} - 2 
	\frac{\mathcal{C}\bullet \alpha_{i}}{\alpha_{i}\bullet \alpha_{i}}\alpha_{i}
	= \mathcal{C} + \left(\mathcal{C}\bullet \alpha_{i}\right) \alpha_{i},\qquad \mathcal{C}\in \operatorname{Pic(\mathcal{X})}.
\end{equation}
Next we need to extend this group by the automorphisms of the Dynkin diagram (that corresponds to some re-labeling of the 
symmetry/surface roots) $\operatorname{Aut}(2A_1^{(1)})\simeq \operatorname{Aut}(D_{6}^{(1)})\simeq \mathbb{D}_{4}$, where
$\mathbb{D}_{4}$ is the usual dihedral group of the symmetries of a square. Generators of $\operatorname{Aut}(2A_1^{(1)})$ can be 
realized as compositions of reflections in some other roots in $\operatorname{Pic}(\mathcal{X})$. We take as generators
\begin{equation*}
	\sigma_{1} = w_{\mathcal{E}_{1} - \mathcal{E}_{3}}\circ w_{\mathcal{E}_{2} - \mathcal{E}_{4}},\qquad
	\sigma_{2} = w_{\mathcal{E}_{1} - \mathcal{E}_{7}}\circ w_{\mathcal{E}_{2} - \mathcal{E}_{8}}\circ
	w_{\mathcal{E}_{3} - \mathcal{E}_{5}}\circ w_{\mathcal{E}_{4} - \mathcal{E}_{6}}\circ 
	w_{\mathcal{H}_{q} - \mathcal{E}_{5} - \mathcal{E}_{6}}\circ w_{\mathcal{H}_{q} - \mathcal{E}_{3} - \mathcal{E}_{4}}
\end{equation*}
that permute the surface and the symmetry roots as follows (here we use the standard cycle notation for permutations):
$\sigma_{1} = (\alpha_{0}\alpha_{3})(\alpha_{1}\alpha_{2}) = (\delta_{0}\delta_{1})$, 
$\sigma_{2} = (\alpha_{0} \alpha_{1}) = (\delta_{0} \delta_{6})(\delta_{1}\delta_{5})(\delta_{2}\delta_{4})$.
The resulting group $\widetilde{W}(2A_1^{(1)})=\operatorname{Aut}(2A_1^{(1)})\ltimes W(2A_1^{(1)})$ is called an  
extended affine Weyl symmetry group and its action on $\operatorname{Pic}(\mathcal{X})$ can 
be further extended to an action on point configurations by elementary birational maps (which lifts to 
isomorphisms $w_{i}: \mathcal{X}_{\mathbf{b}}\to \mathcal{X}_{\overline{\mathbf{b}}}$ on the family of Sakai's surfaces),
this is known as a \emph{birational representation}. We describe it in the following Lemma.

\begin{lemma}\label{lem:bir-rep2A1}
	The action of the generators of the extended affine Weyl group $\widetilde{W}(2A_1^{(1)})$ on some initial point configuration
	\begin{equation*}
	    \left(\begin{matrix}
	    a_0 & a_1\\
	    a_2 & a_3
	    \end{matrix};\ t;\ q,p\right)
	\end{equation*}	
	described using the root variables in the affine $(q,p)$ chart is given by the following birational maps:
	\begin{align*}
	    w_0&:
	    \left(\begin{matrix}
	    -a_0 &2a_0+a_1\\
	    a_2 & a_3
	    \end{matrix};\ t;\ \left(\frac{1}{q} - \frac{a_{0}}{q^{2} p + a_{2} q + t}\right)^{-1}, \right. \\
		&\hskip2in \left.\left(p - \frac{a_{0}(qp + a_{2})}{q^{2} p + a_{2} q + t}\right) \left(1 - \frac{a_{0}q}{q^{2} p + a_{2} q + t}\right) \right),\\		
		w_1&:
		\left(\begin{matrix}
		a_0+2a_1 & -a_1\\  a_2 & a_3
		\end{matrix};\  t;\ q+\frac{a_1}{p-1},p\right),\\
		w_2&: 
		\left(\begin{matrix}
		a_0  & a_1\\
		-a_2 & 2a_2+a_3
		\end{matrix};\  t;\ q+\frac{a_2}{p},p\right),\\		
		w_3&:
		\left(\begin{matrix}
		a_0 &a_1\\
		a_2+2a_3 &-a_3
		\end{matrix};\ t;\ \frac{q (\Upsilon + a_{3} q)}{\Upsilon},
		1+\frac{\Upsilon\left((p-1)\Upsilon - a_{1} a_{3}\right)}{(\Upsilon + a_{3}q)^{2}}\right),
		\intertext{ where $\Upsilon = q^{2} (p-1) + (a_{1} - a_{3}) q  + t$,}
		\sigma_{1}&:
		\left(\begin{matrix}
		a_3 & a_2\\
		a_1 & a_0
		\end{matrix};\ -t;\ -q,1-p\right),\\		
		\sigma_{2}&:
		\left(\begin{matrix}
		a_1 & a_0\\
		a_2 & a_3
		\end{matrix};\quad t;\quad \frac{t}{q},-\frac{q(q p+a_{2})}{t}\right).
	\end{align*}
\end{lemma}

The proof is standard, see \cite{DzhFilSto:2020:RCDOPWHWDPE,DzhTak:2018:SASGTDPE} for similar computations explained in detail. 
Note that although expressions for $w_{0}$ and $w_{3}$ look quite complicated, they can be represented in terms of 
simpler generators as  $w_{0} = \sigma_{2} w_{1} \sigma_{2}$, $w_{3}  = \sigma_{1}  w_{0} \sigma_{1} = \sigma_{1} \sigma_{2} w_{1} \sigma_{2} \sigma_{1}$.


\subsection{Discrete Painlev\'e equations on the $D_{6}^{(1)}$ surface} 
\label{sub:discrete_painlev_e_equations_on_the_d__6_1_surface}

There are two natural simple examples.

\subsubsection{The $\mathbf{[1 \overline{1} \overline{1} 1]}$ Discrete Painlev\'e Equation} 
\label{ssub:dP-KNY}

One standard reference example of a discrete Painlev\'e equation on the $D_{6}^{(1)}$ surface is given in 
\cite{KajNouYam:2017:GAPE}\footnote{In \cite[Equations (8.29--8.30)]{KajNouYam:2017:GAPE} 
$a_{2}$ should in fact be $a_{0}$ and an additional normalization $a_{2}+a_{3}=1$ is missing.} as
\begin{equation}\label{eq:dP-KNY}
\overline{q}+q=-\frac{a_2}{p}-\frac{a_1}{p-1},\qquad p+ \underline{p} = 1+\frac{1-a_{2}-a_{1}}{q}-\frac{t}{q^2},
\end{equation}
with the root variables evolution and normalization as follows
\begin{equation}\label{eq:dP-KNY-evola}
\overline{a}_{0}=a_{0}-1,\quad\overline{a}_{1}=a_{1}+1,\quad \overline{a}_{2}=a_{2}+1,\quad \overline{a}_{3}=a_{3}-1,\qquad a_0+a_1=a_2+a_3=1.
\end{equation}
From the evolution of the root variables \eqref{eq:dP-KNY-evola} we can immediately see that the corresponding translation on the root lattice is
\begin{equation}\label{eq:dP-KNY-trans}
\varphi_*:\upalpha=\langle\alpha_0,\alpha_1,\alpha_2,\alpha_3\rangle\mapsto
\varphi_*(\upalpha)=\upalpha+ \langle 1,-1,-1,1 \rangle\delta,
\end{equation}
and that is why we denote this equation by ${[1 \overline{1} \overline{1} 1]}$. 
The decomposition of the Painlev\'e map $\varphi$ can be written using the generators of $\widetilde{W}(2A_1^{(1)})$ as 
(see \cite{DzhTak:2018:SASGTDPE} on detailed explanation on how to obtain such decompositions)
\begin{equation}\label{eq:dP-KNY-decomp}
\varphi=\sigma_{1} \sigma_{2} \sigma_{1} \sigma_{2} w_{2} w_{1} = \varphi_{2}^{-1}\circ \varphi_{1},
\end{equation}
where partial mappings $\varphi_{1} = w_{2} w_{1}: (q,p)\to (- \overline{q},p)$ and 
$\varphi_{2}= \sigma_{2} \sigma_{1} \sigma_{2} \sigma_{1}: (q,p)\to (-q,\underline{p})$ correspond, up to some gauge transformation, 
to half maps given in \eqref{eq:dP-KNY} (see \cite[Section2.9]{DzhFilSto:2020:RCDOPWHWDPE} for detailed explanations about such gauge ambiguities).
The full map then is 
\begin{equation}\label{eq:dP-KNY-map}
\varphi\left(\begin{matrix}
a_0 & a_1\\
a_2 & a_3
\end{matrix};\ t;\ q,p\right)=\left(\begin{matrix}
a_0-1 & a_1+1\\
a_2+1 & a_3-1
\end{matrix};\ t; \ 
\begin{aligned}
	\overline{q} &=-q-\frac{a_2}{p}-\frac{a_1}{p-1}, \\
 \overline{p} &= 1-p + \frac{1-\overline{a}_{1}-\overline{a_{2}}}{\overline{q}}-\frac{t}{\overline{q}^{2}}
\end{aligned} \right).
\end{equation}
Note that the evolution \eqref{eq:dP-KNY-map} induces the evolution $\varphi_{*}: \operatorname{Pic}(\mathcal{X}) \to \operatorname{Pic}(\overline{\mathcal{X}})$ given by 
	\begin{equation}
		\begin{aligned}
			\mathcal{H}_{q}&\mapsto 5 \overline{\mathcal{H}}_{q} + 2 \overline{\mathcal{H}}_{p} - \overline{\mathcal{E}}_{1234}
			- 2\overline{\mathcal{E}}_{5678}, &\qquad 
			\mathcal{E}_{4}&\mapsto 2\overline{\mathcal{H}}_{q} +  \overline{\mathcal{H}}_{p} - \overline{\mathcal{E}}_{15678},\\
			\mathcal{H}_{p}&\mapsto 2 \overline{\mathcal{H}}_{q} + \overline{\mathcal{H}}_{p} - \overline{\mathcal{E}}_{5678}, &\qquad 
			\mathcal{E}_{5}&\mapsto \overline{\mathcal{H}}_{q}  - \overline{\mathcal{E}}_{8},\\
			\mathcal{E}_{1}&\mapsto 2\overline{\mathcal{H}}_{q} +  \overline{\mathcal{H}}_{p} - \overline{\mathcal{E}}_{45678}, &\qquad 
			\mathcal{E}_{6}&\mapsto \overline{\mathcal{H}}_{q}  - \overline{\mathcal{E}}_{7},\\
			\mathcal{E}_{2}&\mapsto 2\overline{\mathcal{H}}_{q} +  \overline{\mathcal{H}}_{p} - \overline{\mathcal{E}}_{35678}, &\qquad 
			\mathcal{E}_{7}&\mapsto \overline{\mathcal{H}}_{q}  - \overline{\mathcal{E}}_{6},\\
			\mathcal{E}_{3}&\mapsto 2\overline{\mathcal{H}}_{q} +  \overline{\mathcal{H}}_{p} - \overline{\mathcal{E}}_{25678}, &\qquad 
			\mathcal{E}_{8}&\mapsto \overline{\mathcal{H}}_{q}  - \overline{\mathcal{E}}_{5},\\
		\end{aligned}
	\end{equation}
	which gives the expected translation \eqref{eq:dP-KNY-trans} on the symmetry root lattice.


\subsubsection{The $\mathbf{[00 \overline{1}1]}$ Discrete Painlev\'e Equation} 
\label{ssub:dP-Sakai}

A different version of a simple discrete Painlev\'e equation on the $D_{6}^{(1)}$ surface was given in \cite{Sak:2001:RSAWARSGPE}, where this equation 
was called \emph{alt.~d-$\Pain{II}$} equation\footnote{Some minor typos were later corrected in \cite[Remark 2.2]{Sak:2007:PDPETLF}.}. A slightly different
version of the equation was given in \cite[Equations (2.35--2.36)]{Sak:2007:PDPETLF}, this is the version we consider. 
 Using notation from \cite{Sak:2007:PDPETLF}, this equation is given 
as a birational map, written in affine coordinates $(f,g)$ by
\begin{equation}\label{eq:dP-Sakai-map}
\psi\left(\begin{matrix}
a_0 & a_1\\
b_0 & b_1
\end{matrix};\ s;\ f,g\right)=\left(\begin{matrix}
a_0 & a_1\\
b_{0}+1 & b_{1}-1
\end{matrix};\ s; \overline{f}=b_{0} - a_{0} - f - g - \frac{s}{g},\overline{g} = \frac{s}{g} - \frac{b_{0}s}{g \overline{f}}\right),	
\end{equation}
where parameters $a_{i}$ and $b_{i}$ are again the root variables normalized by $a_{0} + a_{1} = b_{0} + b_{1} = 1$. Note that $b_{0}$ and
$b_{1}$ in \cite{Sak:2007:PDPETLF} are the same root variables as $a_{2}$ and $a_{3}$ in \cite{KajNouYam:2017:GAPE}, 
$b_{0} = \chi(\alpha_{2}) = a_{2}$ and $b_{1} = \chi(\alpha_{3}) = a_{3}$, but the base point configuration for this mapping is very 
different from the one shown on Figure~\ref{fig:KNY-pt-conf}. The corresponding translation on the symmetry 
root lattice is
\begin{equation}\label{eq:dP-Sakai-trans}
\psi_*:\upalpha=\langle\alpha_0,\alpha_1,\alpha_2,\alpha_3\rangle\mapsto
\psi_*(\upalpha)=\upalpha+ \langle 0,0,-1,1 \rangle\delta,
\end{equation}
and that is why we denote this equation by ${[00\overline{1}1]}$. To see this we need to look at the geometry of this mapping, 
which is quite different from our model surface on Figure~\ref{fig:KNY-pt-conf}. The base points are 
 \begin{equation}\label{pi-points}
	 \begin{aligned}
		    &\pi_{1}(f=b_{0}-1,g=0),\\
		    &\pi_{2}(F=0,g=0)\leftarrow \pi_{3}(u_2=F=0,v_2=fg=-s)\\
			&\phantom{\pi_{2}(F=0,p=0)}\leftarrow \pi_{4}(u_{3}=F=0,v_{3}=f(fg+s)=a_{0}s),\\
		    &\pi_{5}(f=0,G=0),\\
		    &\pi_{6}(F=0,G=0)\leftarrow \pi_{7}(u_{6}=F=0,v_{6}=fG=-1)\\
			&\phantom{\pi_{6}(F=0,G=0)}\leftarrow \pi_{8}(u_{7}=F=0,v_{7}=f(1+fG)=a_{0} - b_{0}),
	 \end{aligned}
 \end{equation}
(where as usual $F = 1/f$, $G = 1/g$) and their configuration and the resulting surface are shown on Figure~\ref{fig:Sakai-pt-conf}.
\begin{figure}[ht]
	\begin{center}		
	\begin{tikzpicture}[>=stealth,basept/.style={circle, draw=red!100, fill=red!100, thick, inner sep=0pt,minimum size=1.2mm}]
	\begin{scope}[xshift=0cm,yshift=0cm]
	\draw [black, line width = 1pt] (-0.2,0) -- (3.2,0)	node [pos=0,left] {\small $H_{g}$} node [pos=1,right] {\small $g=0$};
	\draw [black, line width = 1pt] (-0.2,3) -- (3.2,3) node [pos=0,left] {\small $H_{g}$} node [pos=1,right] {\small $g=\infty$};
	\draw [black, line width = 1pt] (0,-0.2) -- (0,3.2) node [pos=0,below] {\small $H_{f}$} node [pos=1,above] {\small $f=0$};
	\draw [black, line width = 1pt] (3,-0.2) -- (3,3.2) node [pos=0,below] {\small $H_{f}$} node [pos=1,above] {\small $f=\infty$};
	\node (p1) at (1,0) [basept,label={[xshift = 0pt, yshift=-3pt] \small $\pi_{1}$}] {};
	\node (p2) at (3,0) [basept,label={[xshift = -7pt, yshift=-3pt] \small $\pi_{2}$}] {};
	\node (p3) at (3.5,0.5) [basept,label={[xshift = 0pt, yshift=-3pt] \small $\pi_{3}$}] {};
	\node (p4) at (4.2,0.5) [basept,label={[xshift = 0pt, yshift=-3pt] \small $\pi_{4}$}] {};
	\node (p5) at (0,3) [basept,label={[xshift = 7pt, yshift=-15pt] \small $\pi_{5}$}] {};
	\node (p6) at (3,3) [basept,label={[xshift = -7pt, yshift=-15pt] \small $\pi_{6}$}] {};
	\node (p7) at (3.5,2.5) [basept,label={[xshift = 0pt, yshift=-15pt] \small $\pi_{7}$}] {};
	\node (p8) at (4.2,2.5) [basept,label={[xshift = 0pt, yshift=-15pt] \small $\pi_{8}$}] {};
	\draw [red, line width = 0.8pt, ->] (p3) -- (p2);
	\draw [red, line width = 0.8pt, ->] (p4) -- (p3);
	\draw [red, line width = 0.8pt, ->] (p7) -- (p6);
	\draw [red, line width = 0.8pt, ->] (p8) -- (p7);
	\end{scope}
	\draw [->] (6.5,1.5)--(4.55,1.5) node[pos=0.5, below] {$\operatorname{Bl}_{\pi_{1}\cdots \pi_{8}}$};
	\begin{scope}[xshift=8.5cm,yshift=0cm]
	\draw [blue, line width = 1pt] (-0.2,0) -- (3.7,0)	node [pos=0,left] {\small $H_{g}-K_{12}$} {};
	\draw [blue, line width = 1pt] (0.3,3) -- (3.7,3)	node [pos=0,left] {\small $H_{g}-K_{56}$} {};
	\draw [red, line width = 1pt] (0,-0.2) -- (0,2.7) node [pos=0.5,left] {\small $H_{f}-K_{5}$};
	\draw[red, line width = 1pt] (0.5,-0.2)--(1.5,0.8) node [pos=0,below] {\small $K_{1}$};
	\draw[red, line width = 1pt] (1.2,0.3)--(1.2,3.2) node [pos=1,xshift = 10pt, yshift=7pt] {\small $H_{f} - K_{1}$};
	\draw[blue, line width = 1pt] (3.2,-0.2)--(4.2,0.8) node [pos=0,below] {\small $K_{2}- K_{3}$};
	\draw[blue, line width = 1pt] (3.9,0.1)--(2.9,1.1) node [pos=1,left] {\small $K_{3}- K_{4}$};
	\draw[red, line width = 1pt] (2.7,0.3)--(3.7,1.3) node [pos=0,left] {\small $K_{4}$};	
	\draw[red, line width = 1pt] (-0.2,2.2)--(0.8,3.2) node [pos=0,left] {\small $K_{5}$};
	\draw[blue, line width = 1pt] (4,0.3)--(4,2.7) node [pos=0.5,right] {\small $ H_{f}- K_{26}$};
	\draw[blue, line width = 1pt] (3.2,3.2)--(4.2,2.2) node [pos=0,above] {\small $K_{6}- K_{7}$};
	\draw[blue, line width = 1pt] (3.9,2.9)--(2.9,1.9) node [pos=1,left] {\small $K_{7}- K_{8}$};
	\draw[red, line width = 1pt] (2.7,2.7)--(3.7,1.7) node [pos=0,left] {\small $K_{8}$};	
	\end{scope}
	\end{tikzpicture}
	\end{center}
	\caption{The Sakai Surface for the alt.~d-$\Pain{II}$ equation} 
	\label{fig:Sakai-pt-conf}
\end{figure}	

\begin{lemma}\label{lem:KNY-Sakai}
	On the level of the Picard lattice the basis identification between the model d-P$(2A_1^{(1)}/D_6^{(1)})$ surface 
	on Figure~\ref{fig:KNY-pt-conf} and the surface for the alt.~d-$\Pain{II}$ equation on Figure~\ref{fig:Sakai-pt-conf} is given 
	by 
	\begin{equation}
		\begin{aligned}
				\mathcal{H}_{f} &=2\mathcal{H}_{q} + \mathcal{H}_{p} - \mathcal{E}_{1567}, &\qquad \mathcal{H}_{q} &= \mathcal{H}_{g}, \\
				\mathcal{H}_{g} &=\mathcal{H}_{q}, &\qquad  \mathcal{H}_{p} &= \mathcal{H}_{f} + 2 \mathcal{H}_{g}  - \mathcal{K}_{2678},\\					
				\mathcal{K}_{1} &=\mathcal{E}_{2}, &\qquad  \mathcal{E}_{1} &= \mathcal{H}_{g}  - \mathcal{K}_{2},\\		
				\mathcal{K}_{2} &=\mathcal{H}_{q}-\mathcal{E}_{1}, &\qquad  \mathcal{E}_{2} &= \mathcal{K}_{1},\\		
				\mathcal{K}_{3} &=\mathcal{E}_{3}, &\qquad  \mathcal{E}_{3} &= \mathcal{K}_{3},\\		
				\mathcal{K}_{4} &=\mathcal{E}_{4}, &\qquad  \mathcal{E}_{4} &= \mathcal{K}_{4},\\		
				\mathcal{K}_{5} &=\mathcal{E}_{8}, &\qquad  \mathcal{E}_{5} &= \mathcal{H}_{g}  - \mathcal{K}_{8},\\		
				\mathcal{K}_{6} &=\mathcal{H}_{q}-\mathcal{E}_{7}, &\qquad  \mathcal{E}_{6} &= \mathcal{H}_{g}  - \mathcal{K}_{7},\\		
				\mathcal{K}_{7} &=\mathcal{H}_{q}-\mathcal{E}_{6}, &\qquad  \mathcal{E}_{7} &= \mathcal{H}_{g}  - \mathcal{K}_{6},\\		
				\mathcal{K}_{8} &=\mathcal{H}_{q}-\mathcal{E}_{5}, &\qquad  \mathcal{E}_{8} &= \mathcal{K}_{5}.
		\end{aligned}
	\end{equation}
	
	This gives the following identification between the surface and the symmetry root bases:
	\begin{equation}
		\begin{aligned}
			\delta_{0} &= \mathcal{E}_{1} - \mathcal{E}_{2} &&= \mathcal{H}_{g} - \mathcal{K}_{12},\\
			\delta_{1} &= \mathcal{E}_{3} - \mathcal{E}_{4} &&= \mathcal{K}_{3} - \mathcal{K}_{4},\\
			\delta_{2} &= \mathcal{H}_{q} - \mathcal{E}_{13}&&= \mathcal{K}_{2} - \mathcal{K}_{3},\\
			\delta_{3} &= \mathcal{H}_{p} - \mathcal{E}_{56} &&=  \mathcal{H}_{f}-\mathcal{K}_{26},\\
			\delta_{4} &=\mathcal{E}_{6} - \mathcal{E}_{7} &&= \mathcal{K}_{6} - \mathcal{K}_{7},\\
			\delta_{5} &=\mathcal{E}_{5} - \mathcal{E}_{6} &&= \mathcal{K}_{7} - \mathcal{K}_{8},\\
			\delta_{6} &=\mathcal{E}_{7} - \mathcal{E}_{8} &&= \mathcal{H}_{g} - \mathcal{K}_{56},
		\end{aligned}\qquad 
		\begin{aligned}
		\alpha_{0} &= 2\mathcal{H}_{q} + \mathcal{H}_{p} - \mathcal{E}_{345678} &&= \mathcal{H}_{f} + \mathcal{H}_{g} - \mathcal{K}_{2345},\\
		\alpha_{1} &= \mathcal{H}_{p} - \mathcal{E}_{12} &&= \mathcal{H}_{f} + \mathcal{H}_{g} - \mathcal{K}_{1678},\\
		\alpha_{2} &= \mathcal{H}_{p} - \mathcal{E}_{34} &&= \mathcal{H}_{f} + 2\mathcal{H}_{g} - \mathcal{K}_{234678},\\
		\alpha_{3} &= 2\mathcal{H}_{q} + \mathcal{H}_{p} - \mathcal{E}_{125678} &&= \mathcal{H}_{f}  - \mathcal{K}_{15}.\\
		\end{aligned}
	\end{equation}
	
	The map \eqref{eq:dP-Sakai-map} induces the evolution $\psi_{*}: \operatorname{Pic}(\mathcal{X}) \to \operatorname{Pic}(\overline{\mathcal{X}})$ given by 
	\begin{equation}
		\begin{aligned}
			\mathcal{H}_{f}&\mapsto 3 \overline{\mathcal{H}}_{f} + 2 \overline{\mathcal{H}}_{g} - 2 \overline{\mathcal{K}}_{1}
			- \overline{\mathcal{K}}_{23} - 2\overline{\mathcal{K}}_{5} - \overline{\mathcal{K}}_{67}, &\qquad 
			\mathcal{K}_{4}&\mapsto \overline{\mathcal{H}}_{f} - \overline{\mathcal{K}}_{1},\\
			\mathcal{H}_{g}&\mapsto \overline{\mathcal{H}}_{f} + \overline{\mathcal{H}}_{g} - \overline{\mathcal{K}}_{15}, &\qquad 
			\mathcal{K}_{5}&\mapsto \overline{\mathcal{K}}_{4},\\
			\mathcal{K}_{1}&\mapsto \overline{\mathcal{K}}_{8}, &\qquad 
			\mathcal{K}_{6}&\mapsto \overline{\mathcal{H}}_{f} +  \overline{\mathcal{H}}_{g} -  \overline{\mathcal{K}}_{135},\\
			\mathcal{K}_{2}&\mapsto \overline{\mathcal{H}}_{f} +  \overline{\mathcal{H}}_{g} -  \overline{\mathcal{K}}_{157}, &\qquad 
			\mathcal{K}_{7}&\mapsto \overline{\mathcal{H}}_{f} +  \overline{\mathcal{H}}_{g} -  \overline{\mathcal{K}}_{125},\\
			\mathcal{K}_{3}&\mapsto \overline{\mathcal{H}}_{f} +  \overline{\mathcal{H}}_{g} -  \overline{\mathcal{K}}_{156}, &\qquad 
			\mathcal{K}_{8}&\mapsto \overline{\mathcal{H}}_{f} - \overline{\mathcal{K}}_{5},\\
		\end{aligned}
	\end{equation}
	which gives the expected translation \eqref{eq:dP-Sakai-trans} on the symmetry root lattice.
	
	Finally, the change of variables and parameter correspondence between these two different geometric realizations
	of the $D_{6}^{(1)}$ surface is given by 
	    \begin{equation}\label{eq:KNYtoSakai-coords}
   	 \left\{\begin{aligned}
   	 	f(q,p)&= q(p-1) - \frac{s}{q} + b_{0} - a_{0},\\
   		g(q,p)&= \frac{s}{q},\\
		s&=-t,\quad b_{0}=a_{2},\quad b_{1}=a_{3},\\  
   	 \end{aligned}\right.
    \text{ and conversely, }
    	\left\{\begin{aligned}
   	 	q(f,g)&= -\frac{t}{g},\\
   		p(f,g)&=1 - \frac{g}{t}(f + g + a_{0} - a_{2}),\\
		t &= -s,\quad a_{2} = b_{0}, \quad a_{3} = b_{1}.		
    	\end{aligned}\right.
	    \end{equation}
\end{lemma}

The proof of this Lemma is standard and is omitted. 

\begin{remark} Note that alt.~d-$\Pain{II}$ map $\psi$ can be written using the generators of $\widetilde{W}(2A_1^{(1)})$ as 
$\psi=\sigma_{1} \sigma_{2} \sigma_{1} w_{2}$ and the corresponding mapping, when written in 
coordinates $(q,p)$ of the model surface on Figure~\ref{fig:KNY-pt-conf}, is given by 
\begin{equation}\label{eq:dP-Sakai-map-std}
 \psi\left(\begin{matrix}
 a_0 & a_1\\
 a_2 & a_3
 \end{matrix};\ t;\ q,p\right)=\left(\begin{matrix}
 a_0 & a_1\\
 a_{2}+1 & a_{3}-1
 \end{matrix};\ t; \overline{q}= - \frac{p t}{pq + a_{2}},\overline{p} = \frac{\overline{q}(\overline{q}-a_{1}) + t(p-1)}{\overline{q}^{2}} \right).
\end{equation} 	
Equations \eqref{eq:dP-Sakai-map} and \eqref{eq:dP-Sakai-map-std} are of course related by the change of 
variables \eqref{eq:KNYtoSakai-coords}.
\end{remark}




\bibliographystyle{amsalpha}

\providecommand{\bysame}{\leavevmode\hbox to3em{\hrulefill}\thinspace}
\providecommand{\MR}{\relax\ifhmode\unskip\space\fi MR }
\providecommand{\MRhref}[2]{%
  \href{http://www.ams.org/mathscinet-getitem?mr=#1}{#2}
}
\providecommand{\href}[2]{#2}

\end{document}